\documentclass[preprint]{elsarticle}

\usepackage{caption}
\usepackage{subcaption}

\usepackage[dvipsnames]{xcolor}
\usepackage{lipsum}
\usepackage{amsfonts, amsmath, amssymb}
\usepackage{graphicx}
\usepackage{epstopdf}
\usepackage{algorithmic}

\usepackage[T1]{fontenc}

\usepackage[colorlinks=true]{hyperref}
\usepackage{url}
\usepackage{amsthm}
\usepackage{cleveref}

\ifpdf
  \DeclareGraphicsExtensions{.eps,.pdf,.png,.jpg}
\else
  \DeclareGraphicsExtensions{.eps}
\fi

\usepackage{soul}
\usepackage{bbm}

\usepackage{enumitem}
\usepackage{algorithm,algorithmic,float}
 
\usepackage{tikz}
\usetikzlibrary{positioning}
\usetikzlibrary{backgrounds}

\tikzstyle{arrowstyle} = [->,>=stealth]
\tikzstyle{biarrowstyle} = [<->,>=stealth]

\usepackage{pgfplots}
\usepgfplotslibrary{colormaps}

\colorlet{mygolden}{teal!10}
\colorlet{myblue}{blue}
\colorlet{mygreen}{teal}

\usepackage{enumitem}
\setlist[enumerate]{leftmargin=.5in}
\setlist[itemize]{leftmargin=.5in}

\newcommand{\keyw}[1]{\emph{#1}}

\newdefinition{definition}{Definition}
\newdefinition{example}{Example}
\newdefinition{remark}{Remark}
\newtheorem{mainprob}{Main problem}
\newtheorem{theorem}{Theorem}
\newtheorem{lemma}[theorem]{Lemma}
\newtheorem{proposition}{Proposition}
\newdefinition{corollary}{Corollary}

\usepackage{amsopn}

\newcommand{\repository}{https://github.com/x3042/ExactODEReduction.jl}

\pgfplotsset{
width=6cm, height=5cm, 
compat=1.3,
legend style={font=\scriptsize},
legend image post style={scale=1.0}
}
\pgfplotsset{
    colormap={slategraywhite}{
        rgb255=(112,128,144)
        rgb255=(255,159,101)
    },
    colormap={tealy}{
        color=(teal!50)
        color=(black)
    },
    colormap name=tealy
}
\pgfplotscreateplotcyclelist{marklist}{%
{mark=*},
{mark=o},
{mark=square},
{mark=rectangle}}
\pgfplotscreateplotcyclelist{colorlist}{[colors of colormap={0, 200, 400, 600}]}

\begin{document}

\begin{frontmatter}

\title{Exact hierarchical reductions of dynamical models\\ via linear transformations}

\author[1]{Alexander Demin\corref{cor1}}
\ead{alexander.demin.eternal@gmail.com}

\author[1]{Elizaveta Demitraki}
\ead{egdemitraki@edu.hse.ru}

\author[2]{Gleb Pogudin}
\ead{gleb.pogudin@polytechnique.edu}

\cortext[cor1]{Corresponding author}

\affiliation[1]{
organization={National Research University, Higher School of Economics}, 
city={Moscow}, 
country={Russia}}

\affiliation[2]{
organization={LIX, CNRS, \'Ecole Polytechnique, Institute Polytechnique de Paris},
city={Palaiseau},
country={France}}

\begin{abstract}
Dynamical models described by ordinary differential equations (ODEs) are a fundamental tool in the sciences and engineering. Exact reduction aims at producing a lower-dimensional model in which each macro-variable can be directly related to the original variables, and it is thus a natural step towards the model's formal analysis and mechanistic understanding.
We present an algorithm which, given a polynomial ODE model, computes a longest possible chain of exact linear reductions of the model such that each reduction refines the previous one, thus giving a user control of the level of detail preserved by the reduction.
This significantly generalizes over the existing approaches which compute only the reduction of the lowest dimension subject to an approach-specific constraint.
The algorithm reduces finding exact linear reductions to a question about representations of finite-dimensional algebras.
We provide an implementation of the algorithm, demonstrate its performance on a set of benchmarks, and illustrate the applicability via case studies.
Our implementation is freely available at~\url{\repository}.
\end{abstract}

\begin{keyword}
ordinary differential equations \sep exact reduction \sep lumping \sep dimensionality reduction \sep matrix algebras 
\MSC{34C20, 34-04, 16G10}
\end{keyword}

\end{frontmatter}

\section{Introduction}
Ordinary Differential Equations (ODEs) provide a powerful and expressive language for describing systems evolving in real-time and, thus, are widely used both in the sciences and engineering.
This motivates development of formal methods to analyse the structure of models defined using~ODEs.
One important problem which has been studied actively in the past decade from this angle is \emph{model reduction}~\cite{FeretPNAS,feret2,erode,Cardelli2017a,clue}.

In general, model reduction refers to a variety of techniques aiming at replacing the model of interest with a lower-dimensional one which allows to analyze the dynamics of the original model.
Traditionally, approximate methods such as, e.g., balanced truncation~\cite{antoulas} have been employed.
\emph{Exact model reduction} is a complementary approach in which one lowers the dimension of the model without introducing approximation errors. 
Because of their exactness such reductions are typically connected to the system structure and are, thus, of particular interest in the context of performing formal analysis or deriving mechanistic insights.
Classical examples of exact model reductions include reductions using conservation laws or reductions based on symmetries~\cite{Hubert2013,Camporesi2011,Olver1995}.
In principle, exact transformations can be useful to reduce different measures of complexity of a model besides model dimension.
For example, exact transformations may be used to convert a model to a linear one~\cite{Sankaranarayanan2011}.
Such transformations are beyond the scope of the present paper.

Algorithms for finding exact dimension reductions may be applicable directly to the ODE systems as, for example, the one proposed in the present paper, or to a domain-specific description of a model (which can be then translated to an ODE system).
Rule-based modeling~\cite{Faeder2005,Danos} is a prominent example of such concise domain-specific descriptions for which powerful reduction algorithms have been devised~\cite{FeretPNAS,Camporesi2011,Camporesi2013} and implemented~\cite{feret:CMSB2017}.
Compared to the ODE-based algorithms considered in this paper, algorithms that rely on a rule-based representation can deal with very large models but are applicable only to a class of ODE systems and search only for a class of reductions (via so-called fragments~\cite{FeretPNAS}, see also~\Cref{rem:rules}).

In this paper, we will focus on an important class of exact reductions of ODE models, exact linear lumpings, which correspond to finding a self-contained system of differential equations for a set of \emph{macro-variables} such that each macro-variable is a linear combination of the original variables.
The case when the macro-variables are sought as sums of the original variables has received significant attention, see e.g.~\cite{FeretPNAS,Cardelli2019,feret2, Cardelli2017a}.
In particular, \texttt{ERODE} software has been developed~\cite{erode} which efficiently finds the optimal subdivision of the variables into macro-variables. 
\texttt{CLUE} package~\cite{clue,JimnezPastor2022} was a step towards 
relaxing
these restrictions on the macro-variables.
Compared to the earlier approaches, the macro-variables found by \texttt{CLUE} may involve arbitrary coefficients (not just zeroes and ones as before) and also allow the same variable to appear in several macro-variables at once. Consequently, the dimension of the reduced model in CLUE could be significantly lower~\cite[Table~1]{clue}.
However, the input of the algorithm consisted not only of a model but also of linear functions in the state variables to be preserved (the observables).
Such a set of observables may or may not be available, and guessing it correctly is crucial for finding low-dimensional reductions.
In this paper, we aim at taking the best from both worlds: requiring only a model as input (as \texttt{ERODE}) and allowing the macro-variables to be any linear combinations of the original variables (as \texttt{CLUE}).

The main result of the paper is an algorithm for finding arbitrary exact linear reductions when given only a polynomial ODE model with rational coefficients.
Note that the question of finding such an arbitrary linear lumping of the lowest possible dimension may not be the most meaningful one since any linear first integral yields a reduction of dimension one with constant dynamics.
Instead, our algorithm builds a hierarchy of reductions, more precisely it \emph{finds a longest possible chain of lumpings} in which each reduction refines the next one (for details, see~\Cref{sec:problem}) so that a user can choose the desired level of details to be preserved by reduction by moving along the hierarchy and may find reductions which would likely be missed by~\texttt{ERODE} and~\texttt{CLUE} (e.g., see~\Cref{ex:clueerode,sec:ex_celldeath}).
Such generality comes with a price: our software is typically slower than \texttt{CLUE} and \texttt{ERODE}.

Our algorithm is based on combining the connection of the linear lumping problem to the problem of finding a common invariant subspace of a set of matrices~\cite{LiRabitz,clue} with the structure theory of finite-dimensional algebras.
We use the general framework of existing algorithms over finite~\cite{RONYAI} and algebraically closed~\cite{CGK,speyer} fields and achieve desired efficiency by

\begin{itemize}
    \item sparsity-aware algorithm for finding a basis of an algebra (\Cref{sec:generating-algebra});
    \item exploiting the structure of the input to compute mostly with rational numbers and postponing passing to algebraic number fields as much as possible;
    \item using sparse linear algebra and modular computation 
    to deal with large matrices and expression swell, respectively.
\end{itemize}

We implemented our algorithm, and the implementation is publicly available at~\url{\repository}.
We evaluate its performance on a set of benchmarks from the BioModels database~\cite{BioModels2020}, a large collection of models from life sciences, and demonstrate the produced reduction for two case studies.

The rest of the paper is organized as follows.
In \Cref{sec:problem}, we give a precise definition of exact linear reduction, and formulate explicitly the algorithmic problem we solve in the paper. 
\Cref{sec:algo} contains a detailed description of the algorithm and its justification.
We describe our implementation and report its performance in~\Cref{sec:performance}.
\Cref{sec:examples} contains the case studies describing the reductions produced by our software.


\section{Problem statement}\label{sec:problem}
In the paper, the transpose of a matrix $M$ is denoted by $M^T$.
For a vector $\mathbf{x} = (x_1, \ldots, x_n)$ of indeterminates, by $\mathbb{C}[\mathbf{x}]$ (resp., $\mathbb{R}[\mathbf{x}]$, $\mathbb{Q}[\mathbf{x}]$) we will denote the set of polynomials in $\mathbf{x}$ with complex (resp., real, rational) coefficients.

\begin{definition}[Lumping]\label{def:lumping}
Consider an ODE system
  \begin{equation}\label{eq:main}
    \mathbf{x}' = \mathbf{f}(\mathbf{x}),
  \end{equation}
  in the variables $\mathbf{x} = (x_1, \ldots, x_n)$ with polynomial right-hand side, that is, $\mathbf{f} = (f_1, \ldots, f_n)$ and $f_1, \ldots, f_n \in \mathbb{C}[\mathbf{x}]$.
  We say that a linear transformation $\mathbf{y} = \mathbf{x}L$ with $\mathbf{y} = (y_1, \ldots, y_m)$, $L \in \mathbb{C}^{n \times m}$, and $\operatorname{rank}L = m$ is \keyw{a lumping of~\eqref{eq:main}} if there exists $\mathbf{g} = (g_1, \ldots, g_m)$ with $g_1, \ldots, g_m \in \mathbb{C}[\mathbf{y}]$ such that 
  \[
    \mathbf{y}' = \mathbf{g}(\mathbf{y})
  \]
  for every solution $\mathbf{x}$ of~\eqref{eq:main}. 
  The number $m$ will be called \keyw{the dimension} of the lumping, and the entries of $\mathbf{y}$ will be referred to as \keyw{macro-variables}.

  Note that a lumping is uniquely defined by matrix $L$ only, and the corresponding $\mathbf{g}$ can be computed, see~\Cref{rem:how_get_g} for details.
\end{definition}

Throughout this section, we will work with the following running example~\cite[Example~1]{pozitivizor}.
The ODE model will be derived using the laws of mass-action kinetics~\cite[Ch.~7]{mass_action} from a chemical reaction network with a chemical species $X$ and four more species $A_{UU}$, $A_{UX}$, $A_{XU}$, and $A_{XX}$ which represent a molecule $A$ having two identical binding sites such that each site may be either unbound (denoted by U) or bound to $X$ (denoted by X).
These species satisfy the following reactions 
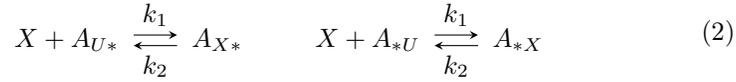
\begin{equation}\label{eq:ex_reactions}
\begin{tikzpicture}[scale=1.0, baseline=(current  bounding  box.center)]
\node (XpAu1) at (-8,0) {$X + A_{U\ast}$};
\node (Ax1) at (-6,0) {$A_{X\ast}$};
\node (XpAu2) at (-4, 0)    {$X + A_{\ast U}$};
\node (Ax2) at (-2,0)    {$A_{\ast X}$};
\draw[arrows={-stealth},transform canvas={yshift=2pt}, shorten <=0.2em, shorten >=0.2em] (XpAu1) to node[above] {$k_1$} (Ax1);
\draw[arrows={-stealth},transform canvas={yshift=-2pt}, shorten <=0.2em, shorten >=0.2em] (Ax1) to node[below] {$k_2$} (XpAu1);
\draw[arrows={-stealth},transform canvas={yshift=2pt}, shorten <=0.2em, shorten >=0.2em] (XpAu2) to node[above] {$k_1$} (Ax2);
\draw[arrows={-stealth},transform canvas={yshift=-2pt}, shorten <=0.2em, shorten >=0.2em] (Ax2) to node[below] {$k_2$} (XpAu2);
\end{tikzpicture}
\end{equation}
where $\ast \in \{X, U\}$, and the reaction rate constants of binding and dissociation are $k_1$ and $k_2$, respectively.
If we denote the concentration of any species $S$ by $[S]$, the corresponding ODE system under the law of mass-action kinetics~\cite[Ch.~7]{mass_action} will be:
  \begin{equation}\label{eq:ex_ode}
      \begin{cases}
        [X]' = k_2([A_{XU}] + [A_{UX}] + 2 [A_{XX}]) - k_1[X] ([A_{XU}] + [A_{UX}] + 2[A_{UU}]),\\
        [A_{UU}]' = k_2([A_{XU}] + [A_{UX}]) - 2 k_1[X] [A_{UU}],\\
        [A_{UX}]' = k_2([A_{XX}] - [A_{UX}]) + k_1[X]([A_{UU}] - [A_{UX}]),\\
        [A_{XU}]' = k_2([A_{XX}] - [A_{XU}]) + k_1[X] ([A_{UU}] - [A_{XU}]),\\
        [A_{XX}]' = k_1[X]([A_{XU}] + [A_{UX}]) - 2 k_2[A_{XX}].
      \end{cases}
\end{equation}

\begin{example}[Conservation laws as lumpings]\label{ex:integral}
We show that the matrix $L = (0
\;1\; 1\; 1\; 1)^T$ yields a lumping of~\eqref{eq:ex_ode}.
We have 
\[
  y = \begin{pmatrix} [X] & [A_{UU}] & [A_{UX}] & [A_{XU}] & [A_{XX}] \end{pmatrix} \cdot L = [A_{UU}] + [A_{UX}] + [A_{XU}] + [A_{XX}].
\]
Using~\eqref{eq:ex_ode} one can check that $y' = 0$, so we can take $g(y) = 0$.
Indeed, $y$ is the total concentration of $A$ and must be constant.
Furthermore, \emph{any linear conservation law} yields a lumping of dimension one.
\end{example}

\begin{example}[More informative lumping]\label{ex:nontrivial}
Another lumping for the same system~\eqref{eq:ex_ode} is given by the matrix
\[
  L = 
  \begin{pmatrix}
    1 & 0 & 0 \\
    0 & 2 & 0 \\
    0 & 1 & 1 \\
    0 & 1 & 1 \\
    0 & 0 & 2
  \end{pmatrix}
  \implies 
  \begin{cases}
    y_1 = [X],\\
    y_2 = 2[A_{UU}] + [A_{UX}] + [A_{XU}],\\
    y_3 = [A_{UX}] + [A_{XU}] + 2[A_{XX}].
  \end{cases}
\]
The macro-variables will satisfy a self-contained system
\[
y_1' = k_2 y_3 - k_1 y_1 y_2, \quad y_2' = k_2 y_3 - k_1 y_1 y_2, \quad  y_3' = -k_2 y_3 + k_1 y_1 y_2.
\]
The rationale behind this reduction is that $y_2$ and $y_3$ are concentrations of unbound and bound sites, respectively.
\end{example}

The above examples demonstrate that one system can have several lumpings (in fact, \eqref{eq:ex_ode} has more), so a natural question is \emph{how to find useful lumpings}.
The state-of-the-art software tools \texttt{CLUE}~\cite{clue} and \texttt{ERODE}~\cite{erode} approach this question by finding the lumping of the smallest dimension satisfying certain constraints:
\begin{itemize}
    \item preserving some quantities of interest unlumped (for \texttt{CLUE}~\cite{clue});
    \item or coming from a partition of the state variables (for \texttt{ERODE}~\cite{erode}).
\end{itemize}
Both constraints may be too restrictive: not all interesting lumpings come from a partition of the state variables (for instance,~\Cref{ex:integral,ex:nontrivial}; see also~\cite[Table~1]{clue}), and it may be complicated to guess in advance meaningful quantities to preserve. 

\begin{example}[Example difficult for \texttt{CLUE} and \texttt{ERODE}, see also~\Cref{sec:ex_celldeath}]\label{ex:clueerode}
    Consider another chemical reaction network~\cite[Eq. (19.20)]{Feinberg} (originally due to Daniel Knight):
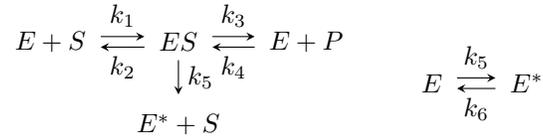
\begin{equation*}
\centering
\begin{tikzpicture}[scale=1.0]
\node (Ep) at (+4.63,-0.55) {$E^*$};
\node (E) at (+3.37,-0.55) {$E$};
\node (EpS) at (-1.7,0)    {$E + S$};
\node (ES) at (0,0)    {$ES$};
\node (EpP) at (1.7,0) {$E + P$};
\node (EppS) at (0,-1.1) {$E^* + S$};

\draw[arrows={-stealth},transform canvas={yshift=2pt}, shorten <=0.2em, shorten >=0.2em] (EpS) to node[above] {$k_1$} (ES);
\draw[arrows={-stealth},transform canvas={yshift=-2pt}, shorten <=0.2em, shorten >=0.2em] (ES) to node[below] {$k_2$} (EpS);

\draw[arrows={-stealth},transform canvas={yshift=2pt}, shorten <=0.2em, shorten >=0.2em] (ES) to node[above] {$k_3$} (EpP);
\draw[arrows={-stealth},transform canvas={yshift=-2pt}, shorten <=0.2em, shorten >=0.2em] (EpP) to node[below] {$k_4$} (ES);

\draw[arrows={-stealth},transform canvas={yshift=2pt}, shorten <=0.2em, shorten >=0.2em] (E) to node[above] {$k_5$} (Ep);
\draw[arrows={-stealth},transform canvas={yshift=-2pt}, shorten <=0.2em, shorten >=0.2em] (Ep) to node[below] {$k_6$} (E);

\draw[arrows={-stealth},transform canvas={yshift=2pt}, shorten <=0.2em, shorten >=0.2em] (ES) to node[right] {$k_5$} (EppS);

\end{tikzpicture}
\end{equation*}

\noindent
where we took the rate constants of $ES \to E^* + S$ and $E \to E^\ast$ to be equal.

As in the case of~\eqref{eq:ex_reactions}, we transform the reactions into an ODE system using the law of mass-action kinetics:
    \[
    \begin{cases}
        [E]' = (k_2 + k_3)[ES] + k_6[E^{*}] - k_1[E][S] - k_4[E][P] - [E],\\
        [S]' = (k_2 + k_5)[ES] - k_1[E][S],\\
        [P]' = k_3[ES] - k_4[E][P],\\
        [ES]' = k_1 [E][S] + k_4 [E][P] - (k_2 + k_3 + k_5)[ES],\\
        [E^*]' = k_5[E] + k_5[ES] - k_6[E^*]
    \end{cases}
    \]
    One meaningful linear reduction is $y = [E] + [ES] - \frac{k_6}{k_5}[E^{*}]$ with the equation $y' = -(k_5 + k_6)y$.
    The macro-variable $y$ can be understood as a potential between the amount of $E$ (typically enzyme), both in the free form $E$ and as a part of the complex $ES$, and $E^{*}$ (typically inactivated enzyme).
    There is a bidirectional flow between these two amounts, so any $k_5$ units of $[E] + [ES]$ will be in an equilibrium with $k_6$ units of $[E^*]$, and, thus, the dynamics of the difference $y$ depends only on itself.
    
    This reduction does not come from a subdivision of the species, so it cannot be found by \texttt{ERODE}.
    Furthermore, finding it using \texttt{CLUE} would require knowing this macro-variable in advance.
    Since the current state of the software does not allow symbolic parameters to appear in the coefficients of the macro-variables (see~\Cref{rem:parameters}), we have found this reduction by taking numerical values of $k_5$ and $k_6$.
\end{example}

An alternative approach would be to find~\emph{all} the lumpings and let the user choose which ones to use. 
The problem is that there may be easily an infinite number of lumpings, for example, similarly to~\Cref{ex:integral}, one can show that the matrix
\[
    L = \begin{pmatrix} \alpha\;\;\; & 1\;\;\; & 1 + \alpha\;\;\; & 1 + \alpha\;\;\; & 1 + 2\alpha \end{pmatrix}^T
\]
yields a lumping of~\eqref{eq:ex_ode} for every number $\alpha$.
Furthermore, as we will explain later, the lumpings are in a bijection with the invariant subspaces of certain matrices coming from the Jacobian of $\mathbf{f}(\mathbf{x})$ and (at least for arbitrary matrices) the invariant subspaces can form an arbitrary algebraic variety~\cite{every_variety}.

The approach we take in this paper is to \emph{find a sequence of reductions} refining each other with the guarantee that this \emph{sequence is of maximal possible length}.

\begin{definition}[Chain of lumpings]\label{def:chain}
  For an ODE system of the form $\mathbf{x}' = \mathbf{f}(\mathbf{x})$, a sequence of linear transformations 
  \[
  \mathbf{y}_1 = \mathbf{x}L_1,\; \mathbf{y}_2 = \mathbf{x}L_2,\;\ldots, \; \mathbf{y}_\ell = \mathbf{x} L_\ell,
  \]
  where $L_1 \in \mathbb{C}^{n \times m_1}, \ldots, L_\ell \in \mathbb{C}^{n \times m_\ell}$,
  is called \keyw{a chain of lumpings} if 
  \begin{enumerate}
      \item $0 < m_1 < \ldots < m_\ell < n$;
      \item $\mathbf{y}_i = \mathbf{x} L_i$ is a lumping of~\eqref{eq:main} for every $1 \leqslant i \leqslant \ell$;
      \item for every $1 < i \leqslant \ell$, there exists a matrix $A_i$ such that $L_{i - 1} = L_i A_i$.
  \end{enumerate}
  The latter means that the reductions given by $L_1, \ldots, L_\ell$ refine each other.
  Such a chain $(L_1, \ldots, L_\ell)$ will be called \keyw{maximal} if it is not contained as a subsequence in any longer chain.  
\end{definition}

We will show (see~\Cref{cor:jordan_hoelder}) that all maximal chains are of the same length, so they are also the longest possible chains.
Given a maximal chain of lumpings, a user can ``zoom in/out'' by going left/right along the chain depending on the desired tradeoff between the size of the reduced model and the amount of information retained.
Thus, we can now formally state the main problem studied in this paper.

\begin{mainprob}
\begin{description}
  \item[]
  \item[Given] a system $\mathbf{x}' = \mathbf{f}(\mathbf{x})$ with $\mathbf{f}$ being a vector of polynomials over $\mathbb{Q}$;
  \item[Compute] a maximal chain of lumpings for the system.
\end{description}
\end{mainprob}

\begin{figure}[H]

\hspace{-8mm}\begin{tikzpicture}
\node[draw,
    rectangle,
    minimum size=0.6cm,
    fill=mygolden
] (orig) at (0,0){\small$\begin{cases}
        [X]' = k_2 ([A_{XU}] + [A_{UX}] + 2 [A_{XX}]) - k_1[X] ([A_{XU}] + [A_{UX}] + 2[A_{UU}]),\\
        [A_{UU}]' = k_2([A_{XU}] + [A_{UX}]) - 2 k_1 [X] [A_{UU}],\\
        [A_{UX}]' = k_2 ([A_{XX}] - [A_{UX}]) + k_1[X]([A_{UU}] - [A_{UX}]),\\
        [A_{XU}]' = k_2 ([A_{XX}] - [A_{XU}]) + k_1[X] ([A_{UU}] - [A_{XU}]),\\
        [A_{XX}]' = k_1[X]([A_{XU}] + [A_{UX}]) - 2 k_2 [A_{XX}].
      \end{cases}$};
 
\node[draw,
    rectangle,
    minimum size=0.2cm,
    fill=mygolden,
    below right = 3cm and -5.55cm of orig
] (red1) at (0,0){\small $\begin{cases}
         y_{4, 1}' = k_2 (y_{4, 3} + 2 y_{4, 4}) - k_1 y_{4, 1}(y_{4, 3} + 2 y_{4, 2}),\\
         y_{4, 2}' = k_2 y_{4, 3} - 2 k_1y_{4, 1} y_{4, 2},\\
         y_{4, 3}' = 2 k_2 (y_{4, 4} - y_{4, 3}) -k_1 y_{4, 1} (y_{4, 3}  - 2 y_{4, 2}),\\
         y_{4, 4}' = k_1 y_{4, 1} y_{4, 3} - 2 k_2 y_{4, 4}
      \end{cases}$};
 
\node[draw,
    rectangle,
    minimum size=0.6cm,
    fill=mygolden,
    below right = 3.25cm and 5.01cm of red1
] (red2) at (0,0){\small$\begin{cases}
         y_{3, 1}' = k_2 y_{3, 3} - k_1 y_{3, 1} y_{3, 2},\\
         y_{3, 2}' = k_2 y_{3, 3} - k_1 y_{3, 1} y_{3, 2},\\
         y_{3, 3}' = -k_2 y_{3, 3} + k_1 y_{3, 1} y_{3, 2}.
      \end{cases}$};
      
\node[draw,
    rectangle,
    minimum size=0.6cm,
    fill=mygolden,
    above right = -2cm and 6.05cm of red2
] (red3) at (0,0){\small$\begin{cases}
         y_{2, 1}' = 0,\\
         y_{2, 2}' = 0.
      \end{cases}$};
      
\node[draw,
    rectangle,
    minimum size=0.6cm,
    fill=mygolden,
    above right = 0.2cm and 6.1cm of red3
] (red4) at (0,0){\small$\begin{cases}
         y_{1, 1}' = 0.
      \end{cases}$};
 
\draw[arrowstyle] (orig.south) + (-2, 0) -- (red1.north)
    node[midway,right]{$\begin{cases}
      y_{4, 1} = [X], \; y_{4, 2} = [A_{UU}],\\
      y_{4, 3} = [A_{UX}] + [A_{XU}],\; y_{4, 4} = [A_{XX}]
    \end{cases}$};
 
 \draw[arrowstyle] (red1.east) -- (red2.west)
    node[midway,below]{\small $\hspace{3mm}\begin{cases}
      y_{3, 1} = y_{4, 1},\\
      y_{3, 2} = y_{4, 3} + 2 y_{4, 2},\\
      y_{3, 3} = y_{4, 3} + 2 y_{4, 4}
    \end{cases}$};
 
\draw[arrowstyle] (red2.north) + (-0.25, 0) -- (red3.south)
    node[midway,right]{\small $\begin{cases}
      y_{2, 1} = y_{3, 2} + y_{3, 3},\\
      y_{2, 2} = y_{3, 1} + y_{3, 3}
    \end{cases}$};
    
\draw[arrowstyle] (red3.north) -- (red4.south)
    node[midway,right]{\small $y_{1, 1} = y_{2, 1}$};
 
\end{tikzpicture}

\caption{Maximal chain of lumpings for~\eqref{eq:ex_ode} and the corresponding reductions}
\label{fig:ex_chain}
\end{figure}
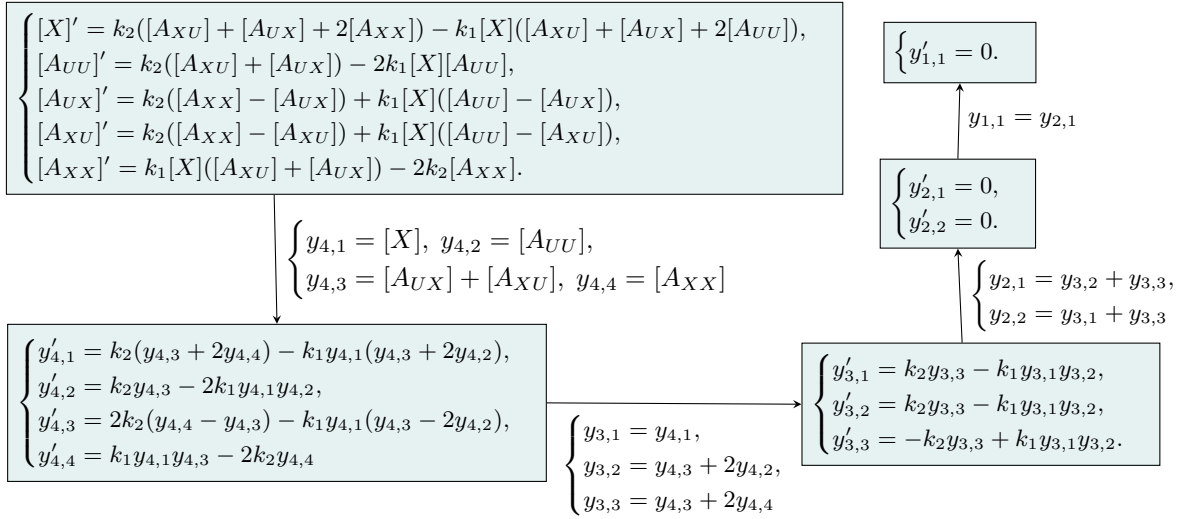

\begin{example}[Maximal chain of lumpings for~\eqref{eq:ex_ode}]\label{ex:chain}~\Cref{fig:ex_chain} shows a chain of lumpings and the corresponding reductions for our example system~\eqref{eq:ex_ode}.
   The blocks contain the reduced systems and the arrows are labeled with the transformations between the consecutive reductions (matrices $A_i$ in the terms of~\Cref{def:chain}).
   This chain of reductions includes our preceding~\Cref{ex:integral,ex:nontrivial} as $\mathbf{y}_1$ and $\mathbf{y}_3$, respectively.
   
   In this example the original dimension $n = 5$ and the dimensions of the reductions are $m_1 = 1,\; m_2 = 2\;, m_3 = 3,\; m_4 = 4$, so this chain is clearly maximal.
\end{example}

\begin{remark}[Connection to symmetries of a rule-based representation]\label{rem:rules}
The model in \eqref{eq:ex_reactions} fits naturally into the framework of rule-based modeling~\cite{Faeder2005,Danos}. 
The model is given by the rules that stipulate the species dynamics, and the rules themselves are essentially chemical reactions parametrized by different values of $\ast \in \{X, U\}$.
Note that these rules are symmetric with respect to the two binding sites of the molecule $A$. 
Thanks to this symmetry, one can find a reduction already at the level of the rules by grouping the species according to the total number of bound sites as follows: 
\[
y_1 = [X], \quad y_2 = [A_{UU}], \quad y_3 = [A_{UX}] + [A_{XU}], \quad y_4 = [A_{XX}].
\]
Note that this is the first reduction in the chain presented on \Cref{fig:ex_chain}.
There exist algorithms that exploit such rule-based symmetries~\cite{Camporesi2013}, and this particular reduction can be found using \texttt{KaDE} software tool~\cite{feret:CMSB2017}\footnote{We thank an anonymous referee for raising our awareness of this fact.}.
The next reduction in the chain, however, is not discovered by \texttt{KaDE}, although it still admits an interpretation in terms of the structure of the underlying reaction network~\cite[Example 1 and Section 3.1]{pozitivizor}.
We see an opportunity for synergy here: our structure-agnostic tool may find new types of reductions, and some of these types can be then incorporated in a scalable rule-based approach.
\end{remark}

\begin{remark}[Connection to quiver-equivariant ODEs]
  One can view a chain of lumpings in the framework of quiver-equivariant dynamical systems~\cite{quivers} with the corresponding quiver being a chain with the maps defined by $A_i$'s on the arrows.
  For this point of view, a natural generalization of the main problem stated above would be to find a maximal (in some sense) quiver $\mathcal{Q}$ such that the original ODE can be represented as $\mathcal{Q}$-equivariant. 
\end{remark}


\section{Algorithm}\label{sec:algo}

For finding a maximal chain of lumpings, we first use theory developed in~\cite{clue} to reduce the problem to a problem about common invariant subspaces of a set of matrices (\Cref{sec:reduction_to_subspaces}) and then solve the new problem using the structure theory of finite-dimensional algebras (\Crefrange{sec:theory}{sec:chain}).
The overall algorithm is summarized in~\Cref{sec:general_algo}.

\subsection{Reduction to the search for common invariant subspaces}\label{sec:reduction_to_subspaces}

Let $\mathbf{x}' = \mathbf{f}(\mathbf{x})$ be an ODE system in variables $\mathbf{x} = (x_1, \ldots, x_n)$ and $\mathbf{f}$ being a row vector of polynomials $f_1, \ldots, f_n \in \mathbb{C}[\mathbf{x}]$.
Let $J(\mathbf{x})$ be the Jacobian matrix of $\mathbf{f}$ with respect to $\mathbf{x}$.
We denote the monomials in $\mathbf{x}$ appearing in $J(\mathbf{x})$ by $m_1(\mathbf{x}), \ldots, m_N(\mathbf{x})$.
Then $J(\mathbf{x})$ can be written uniquely as 
\begin{equation}\label{eq:jacobian_decomposition} 
  J(\mathbf{x}) = \begin{pmatrix}
    \nabla f_1 & \ldots & \nabla f_n
  \end{pmatrix} = \sum\limits_{i = 1}^N J_i\cdot  m_i(\mathbf{x}), \quad\text{ where } \nabla g := \begin{pmatrix}
    \frac{\partial g}{\partial x_1} & \ldots & \frac{\partial g}{\partial x_n}
  \end{pmatrix}^T
\end{equation}
and $J_1, \ldots, J_N$ are constant matrices.

\begin{example}\label{ex:jac_decomp}
   Consider the system
   \[
     x_1' = x_1 - 2 x_2^2,\quad x_2'
      = -x_2 + x_2^2.
   \]
   In this case the decomposition~\eqref{eq:jacobian_decomposition} will be
   \[
     J(x_1, x_2) = \begin{pmatrix} 1 & 0\\ -4x_2 & -1 + 2x_2 \end{pmatrix} = \begin{pmatrix} 1 & 0 \\ 0 & -1 \end{pmatrix} + \begin{pmatrix}0 & 0 \\ -4 & 2 \end{pmatrix} x_2.
   \]
\end{example}

\begin{lemma}\label{lem:lumping_main}
  Using the notation above, the linear transformation $\mathbf{y} = \mathbf{x}L$, where $L \in \mathbb{C}^{n \times m}$, is a lumping of $\mathbf{x}' = \mathbf{f}(\mathbf{x})$ if and only if the column space of $L$ is invariant with respect to $J_1, \ldots, J_N$.
\end{lemma}

\begin{proof}
  For the case $L \in \mathbb{R}^{n \times m}$, the statement follows from~\cite[Lemmas S.I.1 and S.II.1]{clue}.
  The proof of~\cite[Lemmas S.I.1]{clue} remains correct after replacing $\mathbb{R}$ with $\mathbb{C}$, and the proof of~\cite[Lemmas S.II.1]{clue} will be correct for the case of $\mathbb{C}$ if the real inner products are replaced with the complex ones.
\end{proof}

\begin{remark}\label{rem:how_get_g}
    A natural question is, given $L$ satisfying the conditions of~\Cref{lem:lumping_main}, how can we find $\mathbf{g}$ from~\Cref{def:lumping}?
    One approach is the following: we set $\mathbf{y} := \mathbf{x}L$ 
    and choose a subset
    $\widetilde{\mathbf{x}}$ of $\mathbf{x}$ to complete $\mathbf{y}$ to a basis of the linear span of $\mathbf{x}$.
    Then the derivatives $\mathbf{y}'$, which are equal to $\mathbf{f}L$, can be written in terms of $\mathbf{y}$ and $\widetilde{\mathbf{x}}$ via a linear change of coordinates.   \Cref{lem:lumping_main} guarantees that these polynomials will not, in fact, depend on $\widetilde{\mathbf{x}}$ and, thus, will be exactly~$\mathbf{g}$.
    An optimized version of this construction was used already in {\sc CLUE}~\cite{clue}, and is used in our implementation as well.
\end{remark}

\begin{corollary}
   A sequence of linear transformations $\mathbf{y}_1 = \mathbf{x}L_1,\;\ldots, \; \mathbf{y}_\ell = \mathbf{x} L_\ell$ is a chain of lumpings if and only if the column spaces $V_1, V_2, \ldots, V_\ell$ of $L_1, \ldots, L_\ell$ satisfy
   \begin{itemize}
       \item $V_i$ is invariant with respect to $J_1, \ldots, J_N$ for every $1\leqslant i\leqslant \ell$;
       \item $\{0\} \subsetneq V_1 \subsetneq \ldots \subsetneq V_\ell \subsetneq \mathbb{C}^n$.
   \end{itemize}
   Furthermore, the chain of lumpings is maximal if and only if the chain $V_1, \ldots, V_\ell$ is not a subsequence of a chain of subspaces satisfying the two properties above.
\end{corollary}

In order to search for such chains of invariant subspaces, we will use theory of finite dimensional matrix algebras.

\begin{definition}[Matrix algebra]\label{def:algebra}
  Let $k$ be a field (e.g., $k = \mathbb{Q}, \mathbb{R}, \mathbb{C}$).
  A subspace $\mathcal{A} \subseteq k^{n \times n}$ of matrices is called \emph{an algebra} if it is closed under
  multiplication and contains the identity matrix.
  
  For a finite set $A_1, \ldots, A_m \in k^{n \times n}$, we denote the smallest algebra containing $A_1, \ldots, A_m$ by $\langle A_1, \ldots, A_m\rangle$.
  This algebra is equal to the span of all possible products of these matrices.
  
  For an ODE system $\mathbf{x}' = \mathbf{f}(\mathbf{x})$ with $\mathbf{x} = (x_1, \ldots, x_n)$ and $f_1, f_2, \ldots, f_n \in \mathbb{C}[\mathbf{x}]$, we consider the coefficients $J_1, \ldots, J_N$ of the Jacobian matrix of $\mathbf{f}(\mathbf{x})$ written as a polynomial in $\mathbf{x}$ as in~\eqref{eq:jacobian_decomposition}.
  We will call the algebra $\langle I_n, J_1, \ldots, J_N\rangle$ (where $I_n$ is the identity $n\times n$-matrix) \emph{the Jacobian algebra} of the system $\mathbf{x}' = \mathbf{f}(\mathbf{x})$.
\end{definition}

\begin{example}\label{ex:jac_algebra}
  \begin{itemize}
      \item[]
      \item Let $T_n$ be the set of all upper-triangular matrices in $k^{n \times n}$.
   Since the product of two upper-triangular matrices is upper-triangular again, $T_n$ is an algebra.
      \item Consider the system from~\Cref{ex:jac_decomp}. Its Jacobian algebra is
      \[
        \left\langle \begin{pmatrix} 1 & 0 \\ 0 & 1 \end{pmatrix}, \begin{pmatrix} 1 & 0 \\ 0 & -1 \end{pmatrix},  \begin{pmatrix}0 & 0 \\ -4 & 2 \end{pmatrix} \right\rangle = \{M^T \mid M \in T_2\}.
      \]
  \end{itemize}
\end{example}

Since a subspace $V \subset \mathbb{C}^n$ is invariant with respect to $J_1, \ldots, J_N$ if and only if it is invariant with respect to $\langle I_n, J_1, \ldots, J_N\rangle$, that is, invariant w.r.t. any element of the algebra, we will further focus on finding invariant subspaces of this Jacobian algebra.
An immediate benefit is that we can use the Jordan-H\"older theorem~\cite[Theorem~1.5.1]{kir_drozd} to clarify our notion of the maximal chain of lumpings  (\Cref{def:chain}): the definition only requires that a maximal chain cannot be further refined, and this, in general, does not preclude the existence of longer chains. 
The following direct consequence of \cite[Theorem~1.5.1]{kir_drozd} guarantees that a maximal chain indeed has the maximal possible length.

\begin{corollary}\label{cor:jordan_hoelder}
   For a given ODE system $\mathbf{x}' = \mathbf{f}(\mathbf{x})$, all maximal chains of lumpings have the same length.
\end{corollary}


\subsection{Generating the algebra}
\label{sec:generating-algebra}

For performing explicit computations with the Jacobian algebra $\langle I_n, J_1, \ldots, J_N \rangle$ (\Cref{def:algebra}), we will compute its basis. 
\Cref{alg:algebra-basis} gives a simplified version of our approach, which is essentially~\cite[Algorithm~2]{clue} applied to matrices instead of vectors.
Similarly to~\cite{clue}, we employ modular computation (cf.~\cite[Algorithm~3]{clue}) to avoid the intermediate expression swell and use sparse linear algebra.

Building upon this straightforward adaptation of the approach from~\cite{clue}, we significantly improve the performance by taking further advantage of the sparsity of the input and output. 
In the models from literature, the nonlinear part of the model is typically sparse and, as a result, most of $J_1, \ldots, J_N$ are extremely sparse; furthermore, the basis of the Jacobian algebra also can be often chosen to be very sparse.
Hence, many of the matrices $C$ from~\ref{step:alg1-3-b-1} will be sparse as well.
However, some of the matrices computed at the intermediate steps may be still quite dense slowing down the whole algorithm.
We deal with the issue by temporarily deferring~\ref{step:adding} for relatively dense matrices $C$ and then, once the outer loop exits signaling that $P$ is empty, we add each of the deferred matrices to $P$ and restart the iteration.
This way we ensure that we have generated enough sparse matrices in the algebra so that the reductions of the dense matrices will be more sparse now.
Thanks to this optimization, \Cref{alg:algebra-basis} is not a bottleneck in our computation which it was when we used the approach from~\cite{clue} directly.

\begin{algorithm}[H]
\caption{Finding a basis of matrix algebra (basic version)}
\label{alg:algebra-basis}
\begin{description}[itemsep=0pt]
\item[Input] a set of square matrices $A_1, \ldots, A_\ell \in k^{n \times n}$;
\item[Output] a basis $S$ of the smallest linear subspace $\mathcal{A} \subseteq k^{n \times n}$ containing all possible products of $A_1, \ldots, A_\ell$;
\end{description}

\begin{enumerate}[label = \textbf{(Step~\arabic*)}, leftmargin=*, align=left, labelsep=2pt, itemsep=4pt]
    \item Set $S$ to be any basis of the linear span of $A_1, \ldots, A_\ell$ and let $P := S$.
    \item \label{step:alg1-3} While $P \neq \varnothing$ do
    \begin{enumerate}[label = (\alph*), leftmargin=*, align=left, labelsep=2pt, itemsep=4pt]
    \item Take $B$ to be an element of $P$ and remove it from $P$.
    \item For every $A$ in $\{A_1, \ldots, A_\ell\}$ do
    \begin{enumerate}[label = \roman*., leftmargin=*, align=left, labelsep=2pt, itemsep=4pt, ref=\theenumi \theenumii \roman*]
    \item \label{step:alg1-3-b-1} Compute $C := AB$ and reduce $C$ w.r.t. $S$ via Gaussian reduction.
    \item\label{step:adding} If $C \neq 0$, set $S := S \cup \{C\}$ and $P := P\cup \{C\}$.
    \end{enumerate}
    \end{enumerate}
    \item Return $S$.
\end{enumerate}
\end{algorithm}


\subsection{Search for invariant subspaces: algebraic preliminaries}\label{sec:theory}

For this section, we fix a ground field $k$ of characteristic zero.
The cases we are mostly interested in are rational numbers $\mathbb{Q}$, algebraic numbers $\overline{\mathbb{Q}}$, and complex numbers $\mathbb{C}$.

\begin{definition}[Radical of an algebra]\label{def:radical}
Let $\mathcal{A} \subseteq k^{n \times n}$ be an algebra.
\begin{itemize}
    \item A subspace $\mathcal{I} \subseteq \mathcal{A}$ is called \keyw{an ideal} (resp., \keyw{left ideal}) if $AB, BA \in \mathcal{I}$ (resp., $AB \in \mathcal{I}$) for every $A \in \mathcal{A}$ and $B \in \mathcal{I}$.
    \item An ideal (resp., left ideal) $\mathcal{I} \subseteq \mathcal{A}$ is \keyw{nilpotent} if there exists $N$ such that the product of any $N$ elements of $\mathcal{I}$ is zero.
     
    \item The set of all elements $A \in \mathcal{A}$ such that the left ideal $\mathcal{A}\cdot A$ is nilpotent is called \keyw{the radical} of $\mathcal{A}$.
    It is a nilpotent ideal of $\mathcal{A}$ by~\cite[Theorems~3.1.6, 3.1.10]{kir_drozd}.
\end{itemize}
\end{definition}

\begin{example}\label{ex:radical}
   Let $T_n$ be the set of all upper-triangular matrices in $k^{n \times n}$.
   Consider a subset $U_n \subset T_n$ of strictly upper-triangular matrices.
   One can easily verify that $U_n$ is an ideal and the product of any $n$ elements of $U_n$ is zero.
   Since, for every $A \in U_n$, we have $T_n \cdot A \subseteq U_n$, we deduce that $U_n$ is the radical of $T_n$.
\end{example}

Dixon's theorem~\cite[Theorem~11]{constructive_survey} implies that the radical of an algebra $\mathcal{A} \subseteq k^{n \times n}$ can be computed by finding the kernel of a square matrix of order $\dim \mathcal{A} \leqslant n^2$.
The relevance of the notion of radical to our problem is demonstrated by the following lemma.

\begin{lemma}
\label{lemma:radical-kernel}
  Let $\mathcal{A} \subseteq k^{n \times n}$ be an algebra, and let $\mathcal{R} \subset \mathcal{A}$ be its radical.
  If $\mathcal{R} \neq \{0\}$, then the intersection $\bigcap\limits_{R \in \mathcal{R}} \operatorname{Ker} R$ is nonzero and is invariant w.r.t. $\mathcal{A}$.
\end{lemma}

\begin{proof}
   
    Since $\mathcal{R}$ is a nilpotent ideal, there exists the smallest integer $N$ such that the product of any $N$ elements of $\mathcal{R}$ is zero.
   Then, there exists $0 \neq M \in k^{n\times n}$ which is a product of $N - 1$ elements of $\mathcal{R}$.
   Hence, we have $RM = 0$ for every $R \in \mathcal{R}$, so $V := \bigcap\limits_{R \in \mathcal{R}} \operatorname{Ker} R \supseteq \operatorname{Im} M$ is nontrivial.
   
   Consider $v \in V$, $A \in \mathcal{A}$, and $R \in \mathcal{R}$.
   Since $RA \in \mathcal{R}$, we have $RAv = 0$, so $Av \in \operatorname{Ker} R$.
   Thus, $V$ is invariant w.r.t. $\mathcal{A}$.
\end{proof}

\begin{example}
   Consider the system from \Cref{ex:jac_decomp}.
   In~\Cref{ex:jac_algebra}, it was shown that the Jacobian algebra of this system is the set of lower triangular matrices.
   Similarly to \Cref{ex:radical} we find that the radical of this algebra is $\begin{pmatrix}
     0 & 0\\
     \lambda & 0
   \end{pmatrix}$.
   The common kernel of the radical is spanned by the second basis vector yielding the reduction $y' = -y + y^2$ (with $y = x_2$). 
\end{example}

\begin{definition}[Semisimple algebra]
  An algebra $\mathcal{A} \subseteq k^{n \times n}$ is called \keyw{semisimple} if its radical is zero.
\end{definition}

We will use the following characterization of semisimple algebras.

\begin{theorem}[Wedderburn-Artin, {{\cite[Theorems~2.4.3 and~2.6.2]{kir_drozd}}}]\label{thm:wedart}
  Let $\mathcal{A} \subseteq k^{n \times n}$ be a semisimple algebra.
  Then there exist
  \begin{enumerate}
      \item algebras $\mathcal{A}_1 \subseteq k^{n_1 \times n_1}, \ldots, \mathcal{A}_\ell \subseteq k^{n_\ell \times n_\ell}$ such that $\mathcal{A}_i$ does not have a nontrivial proper invariant subspace in $k^{n_i}$ for every $1 \leqslant i \leqslant \ell$,
      \item integers $m_1, \ldots, m_\ell$ such that $n_1m_1 + \ldots + n_\ell m_\ell = n$,
      \item a basis in $k^n$
  \end{enumerate}
  such that, in this basis, we have
  \begin{equation}\label{eq:wedart}
    \mathcal{A} = \left\{ \operatorname{Diag}(\underbrace{A_1, \ldots, A_1}_{m_1 \text{ times}}, \ldots, \underbrace{A_\ell, \ldots, A_\ell}_{m_\ell \text{ times}}) \mid A_1 \in \mathcal{A}_1,\ldots, A_\ell \in \mathcal{A}_\ell \right\},
  \end{equation}
where $\operatorname{Diag}(B_1, \ldots, B_N)$ denotes the block-diagonal matrix with blocks $B_1, \ldots, B_N$.
\end{theorem}

\begin{example}
   Consider the set of all matrices of the form
   \[
   \begin{pmatrix}
     a & b & 0 & 0\\
     -b & a & 0 & 0\\
     0 & 0 & c & 0\\
     0 & 0 & 0 & c
   \end{pmatrix}, \quad \text{where } a, b, c \in \mathbb{Q}.
   \]
   This is a semisimple algebra in the form~\eqref{eq:wedart} with $\ell = 2$, $m_1 = 1$, and $m_2 = 2$.
\end{example}

In the case $\ell = 1$ and $m_1 = 1$ in the decomposition~\eqref{eq:wedart} from~\Cref{thm:wedart}, there are no invariant subspaces in $k^n$ but there still may be invariant subspaces in $\overline{k}^n$ if $k \neq \overline{k}$, where $\overline{k}$ is the algebraic closure of field $k$.
These subspaces can be found using the center of the algebra.

\begin{definition}[Center/Centralizer]\label{def:center}
    Let $\mathcal{A} \subseteq k^{n\times n}$ be an algebra.
    \begin{itemize}
        \item The \keyw{center} of $\mathcal{A}$ is the set of all $M \in \mathcal{A}$ such that $MA = AM$ for every $A \in \mathcal{A}$.
        \item The \keyw{centralizer} of $\mathcal{A}$ is the set of all $M \in k^{n \times n}$ such that $MA = AM$ for every $A \in \mathcal{A}$.
    \end{itemize}
    
\end{definition}

Since, for every fixed $A$, $AM = MA$ is a system of linear equations in the entries, the center and centralizer can be computed by solving a system of linear equations.

\begin{lemma}\label{lem:center}
  Let $\mathcal{A} \subseteq k^{n\times n}$ be a subalgebra.
  Let $\mathcal{C}$ be the centralizer of $\mathcal{A}$.
  For every $C \in \mathcal{C}$, every eigenspace of $C$ is an invariant subspace of $\mathcal{A}$ in $\overline{k}^n$.
\end{lemma}

\begin{proof}
   Let $V$ be an eigenspace of $C$ corresponding to the eigenvalue $\lambda$.
   Let $A \in \mathcal{A}$.
   Then, for $v\in V$, we have $C(Av) = (CA)v = (AC) v = \lambda Av$, so $Av$ belongs to $V$ as well.
\end{proof}

\begin{lemma}\label{lem:equal_blocks_crit}
  Let $\mathcal{A} \subseteq \mathbb{Q}^{n \times n}$ be a semisimple algebra.
  Let $M \in \mathcal{A}$ be a matrix such that the characteristic polynomial of $M$ is of the form $p(t)^d$, where $p(t)$ is $\mathbb{Q}$-irreducible.
  Let $\mathcal{Z}$ and $\mathcal{C}$ be the center and centralizer of $\mathcal{A}$, respectively.
  Then the equality $\dim\mathcal{C} = d^2 \dim\mathcal{Z}$ is equivalent to the fact that, in the Wedderburn-Artin decomposition~\eqref{eq:wedart} of $\mathcal{A}$, we have $\ell = 1$ and $m_1 = d$.
\end{lemma}

\begin{proof}
   We consider the Wedderburn-Artin decomposition~\eqref{eq:wedart} of $\mathcal{A}$.
   For every $1 \leqslant i \leqslant \ell$, we denote the center of $\mathcal{A}_i$ by $\mathcal{Z}_i$.
   Then $\dim \mathcal{Z} = \dim \mathcal{Z}_1 + \ldots + \dim \mathcal{Z}_\ell$.
   The number of irreducible factors of a characteristic polynomial of any element of $\mathcal{A}$ will be at least $m_1 + \ldots + m_\ell$, so $d \geqslant m_1 + \ldots + m_\ell$.
   A direct computation using the Schur's lemma~\cite[Theorem~2.1.1]{kir_drozd} implies that the centralizer $\mathcal{C}$ of $\mathcal{A}$ is isomorphic to $\operatorname{Mat}_{m_1}(\mathcal{Z}_1) \times \ldots \times \operatorname{Mat}_{m_\ell}(\mathcal{Z}_\ell)$, where $\operatorname{Mat}_{m_i}(\mathcal{Z}_i)$ denotes the space of $m_i\times m_i$-block matrices with each block being an element of $\mathcal{Z}_i$ (cf. \cite[Theorem~2.6.4]{kir_drozd}).
   Therefore
   \[
     \dim \mathcal{C} = m_1^2 \dim \mathcal{Z}_1 + m_2^2 \dim \mathcal{Z}_2 + \ldots + m_\ell^2 \dim \mathcal{Z}_\ell.
   \]
   Bounding the right-hand side, we can write
   \[
     \dim\mathcal{C} \leqslant (m_1 + \ldots + m_\ell)^2 (\dim \mathcal{Z}_1 + \ldots + \dim\mathcal{Z}_\ell) \leqslant d^2 \dim \mathcal{Z}
   \]
   Both inequalities will be equalities if and only if $\ell = 1$ and $d = m_1$, and this proves the lemma.
\end{proof}


\subsection{Search for invariant subspaces: how to find one}
In this subsection, we present~\Cref{alg:one_space_full} for finding an invariant subspace if there is any.
The rest of the subsection is devoted to justifying its correctness and termination, see~\Cref{prop:alg_correct}.

\begin{algorithm}
\caption{Finding a nontrivial invariant subspace of an algebra}
\label{alg:one_space_full}
\begin{description}[itemsep=0pt]
\item[Input ] a basis $B_1, \ldots, B_N \in \mathbb{Q}^{n \times n}$ of an algebra $\mathcal{A} \subseteq \mathbb{Q}^{n \times n}$;
\item[Output ] One of the following:
\begin{itemize}
    \item \texttt{NO} if there is no subspace in $\overline{\mathbb{Q}}^n$ invariant w.r.t. $\mathcal{A}$;
    \item nontrivial proper subspace in $\mathbb{Q}^n$ invariant w.r.t. $\mathcal{A}$;
    \item a maximal chain of subspaces in $\overline{\mathbb{Q}}^n$ invariant w.r.t. $\mathcal{A}$.
\end{itemize}
\end{description}

\begin{enumerate}[label = \textbf{(Step~\arabic*)}, leftmargin=*, align=left, labelsep=2pt, itemsep=4pt]
\item[]
    \hspace{-14mm}\emph{Considering corner cases:}
    \item\label{step:full_alg} If $N = n^2$, return \texttt{NO}.
    \item\label{step:full_saturation} For an arbitrary nonzero vector $v$, consider the space $V$ spanned by $B_1v,\ldots, B_Nv$.
    If $\dim V < n$, \textbf{return} $V$.
    
    \smallskip
    \hspace{-14mm}\emph{Examining the radical:}
    \item Find a basis of the radical $\mathcal{R}$ of $\mathcal{A}$ (\Cref{def:radical}) using Dixon's theorem~\cite[Theorem~11]{constructive_survey}.
    \item\label{step:full_radical} If $\dim \mathcal{R} > 0$ compute the common kernel of the basis elements of $\mathcal{R}$ and \textbf{return} it (see~\Cref{lemma:radical-kernel}).
    
    \smallskip
    \hspace{-14mm}\emph{Semisimple case:}
    \item\label{step:full_semisimple} Set $D := 1$.
    \item\label{step:full_sampling1} Compute $M := \sum_{i = 1}^N a_i B_i$, where $a_1, \ldots, a_N$ are sampled independently and uniformly at random from $\{1, 2, \ldots, D\}$.
    \item\label{step:reducible} If the characteristic polynomial of $M$ has at least two distinct $\mathbb{Q}$-irreducible factors (say, $p_1(t)$ and $p_2(t)$):
    \begin{enumerate}[ref= \theenumi (\alph*)]
        \item Check the invariance of $\operatorname{Ker}p_1(M)$ w.r.t. $B_1, \ldots, B_N$.
        \item\label{step:full_decomposable} If it is invariant, \textbf{return} $\operatorname{Ker}p_1(M)$. Otherwise, set $D := 2D$ and \textbf{go to}~\ref{step:full_sampling1}.
    \end{enumerate}
    \item Write the characteristic polynomial of $M$ as $p(t)^d$, where $p(t)$ is $\mathbb{Q}$-irreducible.
    \item\label{step:full_center} Compute the center $\mathcal{Z}$ and centralizer~$\mathcal{C}$ of $\mathcal{A}$ (\Cref{def:center}).
    \item If $\dim \mathcal{C} < d^2 \dim\mathcal{Z}$, set $D := 2D$ and \textbf{go to}~\ref{step:full_sampling1}.
    \item\label{step:sampling_center} Let $C_1, \ldots, C_s$ be a basis of $\mathcal{C}$.
    Set $C := \sum_{i = 1}^s b_i C_i$, where $b_1, \ldots, b_s$ are sampled independently and uniformly at random from $\{1, 2, \ldots, D\}$. 
    \item\label{step:full_compute_q} Compute $q(t)$, the minimal polynomial of $C$. If $q$ is $\mathbb{Q}$-reducible or $\deg q < d \dim \mathcal{Z}$, set $D := 2D$ and \textbf{go to}~\ref{step:full_sampling1}.
    \item\label{step:eigenspaces} Let $V_1, \ldots, V_\ell$\\  
    (where $\ell = d \dim \mathcal{Z}$) be the eigenspaces of $C$.
    \item\label{step:full_return_chain} \textbf{Return} $V_1 \subset V_1\oplus V_2 \subset \ldots \subset V_1 \oplus V_2 \oplus\ldots \oplus V_{\ell - 1}$.
\end{enumerate}
\end{algorithm}

\begin{proposition}\label{prop:chpoly}
    Let $\mathcal{A} \subseteq \mathbb{Q}^{n \times n}$ be a semisimple algebra such that there are no nontrivial proper $\mathcal{A}$-invariant subspaces in $\mathbb{Q}^n$.
    Let $M_1, \ldots, M_N$ be a basis of $\mathcal{A}$ as a vector space.
    Then the polynomial
    \[
      \operatorname{det} (x_1M_1 + \ldots + x_N M_N) \in \mathbb{Q}[x_1, \ldots, x_N]
    \]
    is of the form $P^d$, where $P$ is irreducible over $\mathbb{Q}$.
\end{proposition}

\begin{remark}[On the importance of being a basis]
  While the statement of \Cref{prop:chpoly} may sound quite natural, the situation is in fact quite subtle: if one replaces linear basis with a set of generators of $\mathcal{A}$ in the statement of the proposition, it will not longer be true~\cite[Theorem 1.2 and Subsection 2.1]{KlepVolvic}.
\end{remark}

\begin{proof}[Proof of~\Cref{prop:chpoly}]
    By performing a change of coordinates over $\mathbb{Q}$, we will assume that $M_1$ is the identity matrix.

    Let $\overline{\mathcal{A}}$ be the complexification of $\mathcal{A}$.
    By the Wedderburn-Artin theorem~\cite[Corollary 2.4.4]{kir_drozd}, there exist $n_1, \ldots, n_\ell$ such that $N = n_1^2 + \ldots + n_\ell^2$ and \begin{equation}\label{eq:decom_alg}
        \overline{\mathcal{A}} \cong \operatorname{Mat}_{n_1}(\mathbb{C}) \times \ldots \times \operatorname{Mat}_{n_\ell}(\mathbb{C}).
    \end{equation}

    Then the complexification $\mathbb{C}^n$ of the original representation $\mathbb{Q}^n$ of $\mathcal{A}$ can be decomposed~\cite[Theorem~2.6.2]{kir_drozd} as
    \begin{equation}\label{eq:decomp_module}
        \mathbb{C}^n = k_1V_{1} \oplus k_2V_2 \oplus \ldots \oplus k_\ell V_\ell,
    \end{equation}
    where $V_i \cong \mathbb{C}^{n_i}$ is the unique irreducible representation of $\operatorname{Mat}_{n_i}(\mathbb{C})$.
    We denote the base change corresponding to~\eqref{eq:decomp_module} by $C \in \mathbb{C}^{n\times n}$.
    Then $CMC^{-1}$, where $M := x_1M_1 + \ldots + x_NM_N$, is block diagonal with the dimensions of blocks as in~\eqref{eq:decomp_module}. Since $M_1, \ldots, M_N$ span the whole $\overline{\mathcal{A}}$, the distinct nonzero entries of $CMC^{-1}$ are $\mathbb{C}$-linearly independent linear forms in $x_1, \ldots, x_N$. 
    By denoting these forms by $y_1, \ldots, y_N$ we obtain an
     invertible matrix $B \in \mathbb{C}^{N\times N}$ such that $\mathbf{y} := B\mathbf{x}$, and, reordering $y_1, \ldots, y_N$ if necessary, one has
    \[
      CMC^{-1} = \operatorname{diag}(\underbrace{Y_1, \ldots, Y_1}_{k_1 \text{times}}, \ldots, \underbrace{Y_\ell, \ldots, Y_\ell}_{k_\ell \text{times}}),
    \]
    where $Y_i$ is a matrix with entries $y_{n_1 + \ldots + n_{i - 1}^2 + 1}, \ldots, y_{n_1^2 + \ldots + n_i^2}$ for $1 \leqslant i \leqslant \ell$.

    Then we have
    \[
      \operatorname{det}(M) = \operatorname{det}(CMC^{-1}) = \operatorname{det}(Y_1)^{k_1} \ldots \operatorname{det}(Y_\ell)^{k_\ell}.
    \]
    Furthermore, since $M_1$ is the identity, $\det(M)|_{x_1 = x_1 + t}$ as a polynomial in $t$ is the characteristic polynomial of $-M$.
    Let $Q(\mathbf{x}) := \det Y_1\ldots \det Y_\ell \in \mathbb{Q}[\mathbf{x}]$.
    Then $Q|_{x_1 = x_1 + t}$ as a polynomial in $t$ is the minimal polynomial of $-M$.
    
    Since $\det Y_i$ is a determinant of a matrix with independent entries, it is irreducible over $\mathbb{C}$.
    Let $p(\mathbf{x})$ be a $\mathbb{Q}$-irreducible divisor of $\det M$.
    Then $p$ divides $Q$, so, by reordering $Y_i$'s if necessary, we can assume that $p(\mathbf{x}) = \det Y_1\ldots \det Y_r$ for $r \leqslant \ell$.
    Assume that $r < \ell$.
    Set $p_0(t) := p(x_1 - t, x_2, \ldots, x_N)$ and consider $p_0(M)$.
    We will have 
    \begin{multline*}
      C p_0(M) C^{-1} = \\ \operatorname{diag}(\underbrace{0, \ldots, 0}_{k_1 + \ldots + k_r \text{ times}}, \underbrace{p_0(Y_{r + 1}), \ldots, p_0(Y_{r + 1})}_{k_{r + 1} \text{ times}}, \ldots, \underbrace{p_0(Y_{r + 1}), \ldots, p_0(Y_{r + 1})}_{k_\ell \text{ times}}).
    \end{multline*}
    Since $p_0$ is coprime with the characteristic polynomials of $Y_{r + 1}, \ldots, Y_\ell$, each of the matrices $p_0(Y_{r + 1}), \ldots, p_0(Y_\ell)$ is nonsingular.
    Therefore, the kernel of $C p_0(M) C^{-1}$ is exactly the span of the first $k_1 + \ldots + k_r$ basis vectors.
    Therefore, the kernel of $p_0(M)$ is the span of this many first columns of $C^{-1}$.
    Therefore, the kernel of $p_0(M)$ is $\overline{\mathcal{A}}$-invariant and is defined over $\mathbb{C}$.
    On the other hand, the entries of $p_0(M)$ belong to $\mathbb{Q}(\mathbf{x})$, so the kernel of $p_0(M)$ in fact is defined over $\mathbb{C} \cap \mathbb{Q}(\mathbf{x}) = \mathbb{Q}$.
    Therefore, the kernel of $p_0(M)$ yields a nontrivial $\mathcal{A}$-invariant subspace of $\mathbb{Q}^n$ contradicting with the irreducibility of this representation.
    Hence $p$ must be equal to $Q$, so $\det M$ must be a power of $p$.
\end{proof}

The proof of the proposition provides a way to find the degree of $\deg P$.

\begin{corollary}\label{cor:Pdegree}
    In the notation of the proof (see~\eqref{eq:decom_alg}) of~\Cref{prop:chpoly}, $\deg P = n_1 + n_2 + \ldots + n_\ell$.
\end{corollary}

\begin{proposition}\label{prop:alg_correct}
  \Cref{alg:one_space_full} is correct and terminates with probability one.
\end{proposition}

\begin{remark}[On the probability of termination]
    By ``terminates with probability one'' we mean that the algorithm makes a random choice in an infinite probability space, and will terminate for all the choices except for a set of probability zero. 
    Simply put, the algorithm repeatedly tosses a coin and will terminate as long as there will be at least one heads which will eventually happen with probability one.
\end{remark}

\begin{proof}[Proof of~\Cref{prop:alg_correct}]
   We will first prove the \emph{correctness}.
   If the algorithm returned on~\ref{step:full_alg}, then $\mathcal{A}$ is the full matrix algebra, and does not have any nontrivial proper invariant subspace.
   If the algorithm returned on~\ref{step:full_saturation}, then the returned subspace is invariant by construction.
   If the algorithm returned on~\ref{step:full_radical}, the returned subspace is nonzero, proper (since if $\mathcal{A}$ was a zero algebra, the algorithm would have returned at~\ref{step:full_saturation}), and invariant due to~\Cref{lemma:radical-kernel}.
   
   It remains to consider the case when the algorithm returns after~\ref{step:full_semisimple}.
   If the algorithm returned on~\ref{step:full_decomposable}, then the returned subspace is invariant by construction and is nonzero because $p_1(t)$ divides the charpoly of $M$, so $p_1(M)$ is a singular matrix.
   Finally, consider the case when the algorithm returned on~\ref{step:full_return_chain}.
   Consider the decomposition~\eqref{eq:wedart} from~\Cref{thm:wedart} for $\mathcal{A}$.
   If the algorithm reached~\ref{step:sampling_center}, it contains a matrix with the charpoly being $p(t)^d$ with $\mathbb{Q}$-irreducible $p(t)$ such that $\dim \mathcal{C} = d^2 \dim \mathcal{Z}$.
   \Cref{lem:equal_blocks_crit} implies that, in the decomposition~\eqref{eq:wedart}, we have $\ell = 1$ and $m_1 = d$.
   Thus, the whole space $\mathbb{Q}^n$ can be written as $U_1 \oplus U_2 \oplus \ldots \oplus U_d$ such that each of $U_i$'s is $\mathcal{A}$-invariant without proper nontrivial $\mathcal{A}$-invariant subspaces.
   \cite[Corollary~2.2.4]{kir_drozd} implies that, over $\overline{\mathbb{Q}}$, each of $U_i$'s can be decomposed as a direct sum of at most $\dim\mathcal{Z}$ $\mathcal{A}$-invariant subspaces.
   Therefore, the whole $\overline{\mathbb{Q}}^n$ can be decomposed into at most $d \dim\mathcal{Z}$ $\mathcal{A}$-invariant subspaces by~\cite[Theorem~2.6.2]{kir_drozd}.
   
   \Cref{lem:center} implies that each of $V_i$'s from~\ref{step:eigenspaces} is an invariant subspace w.r.t. $\mathcal{A}$.
   Since there are $d \dim\mathcal{Z}$ of them, each of $V_i$'s does not contain nontrivial proper $\mathcal{A}$-invariant subspaces.
   Therefore, the chain $V_1 \subset V_1\oplus V_2 \subset \ldots \subset V_1 \oplus V_2 \oplus\ldots \oplus V_{\ell - 1}$ returned at~\ref{step:full_return_chain} is maximal.
   This finished the proof of the \emph{correctness} of the algorithm.
   
   We will now prove that the algorithm \emph{terminates} with probability one.
   Consider the decomposition~\eqref{eq:wedart} of $\mathcal{A}$ from~\Cref{thm:wedart}.
   Consider variables $z_1, \ldots, z_N$ and a matrix $M_0 := \sum_{i = 1}^N z_i B_i$.
   Then $M$ at~\ref{step:full_sampling1} is a specialization of $M_0$ at $z_i = a_i$.
   Let $P(z_1, \ldots, z_N, t)$ be the charpoly of $M_0$.
   Consider the decomposition~\eqref{eq:wedart} for $\mathcal{A}$.
   For every $1 \leqslant i \leqslant \ell$, we apply~\Cref{prop:chpoly} to the block corresponding to $\mathcal{A}_i$ and obtain a $\mathbb{Q}$-irreducible $P_i$ and its power $d_i$.
   Thus, we obtain the following factorization for $M_0$
   \[
     P = P_1^{d_1 m_1} P_2^{d_2 m_2}\ldots P_\ell^{d_\ell m_\ell}.
   \]
   The characteristic polynomial of $M$ computed at~\ref{step:full_sampling1} equals $P(a_1, \ldots, a_N, t)$.
   \ul{Assume that $P_i(a_1, \ldots, a_N, t)$ is $\mathbb{Q}$-reducible for every $1 \leqslant i \leqslant s$ and these polynomials are distinct.}
   \begin{itemize}
       \item Assume that $\ell > 1$.
       Then $p_1(t)$ from~\ref{step:reducible} will be equal to $P_i(a_1, \ldots, a_N, t)$ for some $i$.
       Then $\operatorname{Ker} p_1(M)$ will be the subspace corresponding to the $i$-th block in the decomposition~\eqref{eq:wedart}.
       The subspace is invariant, so it will be returned on~\ref{step:full_decomposable}.
       \item Assume that $\ell = 1$.
       We will study matrix $C$ similarly to the way we studied $M$ above. 
       Let $y_1, \ldots, y_s$ be independent variables, and we define $C_0 := y_1 C_1 + \ldots + y_s C_s$.
       By the same argument as in the proof of~\Cref{lem:equal_blocks_crit}, we have $\mathcal{C} \cong \operatorname{Mat}_r(\mathcal{Z})$ for some integer $r$.
       By~\cite[Proposition~2.3.4]{kir_drozd}, algebra $\mathcal{C}$ is simple and every $\mathcal{C}$-module (in particular, our ambient space $\mathbb{Q}^n$) is a direct sum of isomorphic copies of the same $\mathcal{C}$-module.
       We apply~\Cref{prop:chpoly} to this module and deduce that the characteristic  polynomial of $C_0$ is of the form $Q(y_1, \ldots, y_s, t)^h$ for some integer $h$ and $\mathbb{Q}$-irreducible polynomial $Q$.
       Furthermore, $\deg Q = d \dim \mathcal{Z}$ by~\Cref{cor:Pdegree}.
       \ul{Assume that $Q(b_1, \ldots, b_s, t)$ is $\mathbb{Q}$-irreducible.}
       Then $Q(b_1, \ldots, b_s, t)$ will be the minimal polynomial of $C$, so this polynomial will not satisfy the condition of~\ref{step:full_compute_q} and, thus, the algorithm will terminate without going back to~\ref{step:full_sampling1}.
   \end{itemize}
   
   Combining the two underlined assumptions in the text above, we see that the algorithm will return for a fixed value of $D$ if the following conditions hold:
   \begin{enumerate}
       \item $P_i(a_1, \ldots, a_N, t)$ is $\mathbb{Q}$-reducible for every $1 \leqslant i \leqslant s$ and these polynomials are all distinct;
       \item $Q(b_1, \ldots, b_s, t)$ is $\mathbb{Q}$-irreducible.
   \end{enumerate}
   \cite[Theorem~2.1]{Cohen} implies that there exists constants $C_0, C_1$ such that the probability of any of $P_i(a_1, \ldots, a_N, t)$'s and $Q(b_1, \ldots, b_s, t)$ being $\mathbb{Q}$-reducible is less that $\frac{C_1}{\sqrt[3]{D}}$ if $D > C_0$.
   Furthermore, the Schwartz-Zippel lemma~\cite[Proposition~98]{Zippel} implies that there exists a constant $C_2$ such that the probability of any of $P_i(a_1, \ldots, a_N, t)$'s being equal does not exceed $\frac{C_2}{D}$.
   Therefore, for $D > C_0$, the probability that $D$ will be updated is at most $\frac{C_1}{\sqrt[3]{D}} + \frac{C_2}{D}$.
   This number will eventually become less than $0.99$, so the probability of non-termination will be bounded by $0.99 \cdot 0.99 \cdot \ldots = 0$.
\end{proof}
\subsection{Search for invariant subspaces: how to find a chain}\label{sec:chain}

In this section, we describe how to use~\Cref{alg:one_space_full} in a recursive manner to find a maximal chain of invariant subspaces in $\overline{\mathbb{Q}}$ w.r.t. the Jacobian algebra $\mathcal{A} \subset \mathbb{Q}^{n \times n}$ of an ODE system.
We will denote a basis of $\mathcal{A}$ by $B_1, \ldots, B_N$

In the cases when~\Cref{alg:one_space_full} applied to $B_1, \ldots, B_N$ returned \texttt{NO} or a maximal chain of invariant subspaces, we are done.
Therefore, we consider the case when~\Cref{alg:one_space_full} returns a single invariant subspace $V \subset \mathbb{Q}^n$.
In this case, we consider two subproblems:
\begin{enumerate}
    \item \emph{Restriction.} Since $V$ is invariant w.r.t. $B_1, \ldots, B_N$, there are well-defined restrictions $B_1|_V, \ldots, B_N|_V$.
    We fix a basis in $V$ and will denote the matrix representations for these restricted operators also by $B_{1}^\ast, \ldots, B_{N}^\ast$.
    
    \item \emph{Quotients.} Consider the quotient space $\mathbb{Q}^n / V$ and the quotient map $\pi\colon \mathbb{Q}^n \to \mathbb{Q}^n / V$ (see~\cite[3.83, 3.88]{linalg}).
    Since $V$ is invariant w.r.t. $B_1, \ldots, B_N$, we can consider the quotient operators $B_1 / V, \ldots, B_N / V$ (see~\cite[5.14]{linalg}), we denote their matrix representations by $B_1^\circ, \ldots, B_N^\circ$.
    Note that, for every their common invariant subspace $U \subset \mathbb{Q}^n / V$, the subspace $\pi^{-1}(U) \subset \mathbb{Q}^n$ is invariant w.r.t. $B_1, \ldots, B_N$. 
\end{enumerate}
Note that the aforementioned matrix representations can be computed solving linear systems in $n$ variables.
Thus, we can work recursively with algebras $\langle B_1^\ast, \ldots, B_N^\ast \rangle$ on $V$ and $\langle B_1^\circ, \ldots, B_N^\circ \rangle$ on $\mathbb{Q}^n / V$.
If the resulting maximal chains of invariant subspaces are
\[
  0 \subsetneq V_1 \subsetneq \ldots \subsetneq V_s \subsetneq V \quad\text{ and }\quad 0\subsetneq U_1 \subsetneq \ldots \subsetneq U_r \subsetneq \mathbb{Q}^n / V,
\]
then we can return the following maximal chain of invariant subspaces for $B_1, \ldots, B_N$
\[
  0 \subsetneq V_1 \subsetneq \ldots \subsetneq V_s \subsetneq V \subsetneq \pi^{-1}(U_1) \subsetneq \ldots \subsetneq \pi^{-1}(U_r) \subsetneq \mathbb{Q}^n.
\]


\subsection{Putting everything together}\label{sec:general_algo}

In this section we collect the subroutines from the preceding sections into the complete algorithm for finding a maximal chain of lumpings.

\begin{algorithm}[H]
\caption{Finding a maximal chain of lumpings}
\label{alg:main}
\begin{description}[itemsep=0pt]
\item[Input ] ODE system $\mathbf{x}' = \mathbf{f}(\mathbf{x})$ with $\mathbf{x} = (x_1, \ldots, x_n)$, $\mathbf{f} = (f_1, \ldots, f_n) \in \mathbb{Q}[\mathbf{x}]^n$;
\item[Output ] a maximal chain of lumpings (see~\Cref{def:chain} and~\Cref{ex:chain});
\end{description}

\begin{enumerate}[label = \textbf{(Step~\arabic*)}, leftmargin=*, align=left, labelsep=2pt, itemsep=4pt]
    \item Compute the Jacobian $J(\mathbf{x})$ of $\mathbf{f}$ and the matrices $J_1, \ldots, J_\ell \in \mathbb{Q}^{n \times n}$ from its decomposition as in~\eqref{eq:jacobian_decomposition}.
    \item \label{step:algmain:2} Use~\Cref{alg:algebra-basis} to compute the basis $B_1, \ldots, B_N$ of the Jacobian algebra $\mathcal{A} = \langle I_n, J_1, \ldots, J_\ell \rangle$ of the system.
    \item \label{step:algmain:3} Apply~\Cref{alg:one_space_full} in a recursive way as decribed in~\Cref{sec:chain} to compute a maximal chain $V_1 \subsetneq \ldots \subsetneq V_s$ of subspaces in $\overline{\mathbb{Q}}^n$ invariant w.r.t. $\mathcal{A}$.
    \item For each $1 \leqslant i \leqslant s$, find a matrix $L_i$ with the columns being a basis of $V_i$.
    \item \textbf{Return} $L_1, \ldots, L_s$.
\end{enumerate}
\end{algorithm}


\section{Implementation and performance}\label{sec:performance}

We have implemented \Cref{alg:main} (and all the algorithms it relies on) in Julia language \cite{bezanson2017julia} as a part of \texttt{ExactODEReduction.jl} package. The package together with relevant resources to replicate our results is freely available at
\begin{center}
    \url{\repository}
\end{center}

We use libraries \texttt{AbstractAlgebra.jl}
and \texttt{Nemo.jl}~\cite{Nemo.jl-2017}. Internally, this results in using \texttt{FLINT}~\cite{Hart2010} and \texttt{Calcium}~\cite{Calcium-arxiv} (for complex number arithmetic). 
We use a version of the code from~\cite{pozitivizor} to improve interpretability of the computed lumpings. 
Additionaly, during the development stage, various components of the package were profiled on collections of sparse matrices from the SuiteSparse dataset~\cite{suitesparse}.
Our implementation accepts models typed manually or from the files in the \texttt{ERODE} *.ode format~\cite[Section 3.2]{erode}.
We provide documentation, 
installation instructions, and usage examples.

\begin{remark}[Encoding scalar parameters]\label{rem:parameters}
    Models in the literature often involve scalar parameters. 
    Our algorithm transforms each such parameter $k$ into a state $k(t)$ satisfying equations $k'(t) = 0$ 
    (same transformation is used in \texttt{ERODE} under the name ``currying''~\cite[p. 15]{erode_manual}).
    This way, one may also find reductions in the parameters space, not only in the state space.
    Another approach to handle the parameters could be to adjoin them to the field of coefficients (as allowed in {\sc CLUE}) thus allowing parameters in the coefficients of a reduction.
    This feature is not implemented at the moment but most of the presented theory and algorithms can be reused for this case.
\end{remark}

We will demonstrate the performance of our implementation on a set of benchmarks\footnote{Models are available at~\url{\repository/tree/main/data/ODEs}, commit hash \texttt{678d32c5bbc8beedc9e22b673238cde0ec673a46}.}. 
We use benchmarks from the BioModels database~\cite{BioModels2020} collected in~\cite{largescale} of dimensions ranging from $4$ to $133$. We run \Cref{alg:main} over rationals on each of the models. \Cref{table:benchmarks} contains benchmark results aggregated by models' dimension. For each range, we report:
\begin{itemize}
    \item the number of models considered;
    \item the (average) length of a chain of reductions found;
    \item 
    the (average) number of nonequivalent reductions, where equivalence is taken up to adding states with constant dynamics.
    We have chosen to report this because we think is it a reasonable first approximation to the number of 
    interesting reductions;
    \item the (minimum, average, maximum) elapsed runtime of our implementation;
\end{itemize}

\begin{table}[!htbp]
    \setlength{\tabcolsep}{4pt}
  \centering
    \begin{tabular}{cc|cc|ccc}
    \hline
    \multicolumn{2}{c}{Models info} &
    \multicolumn{2}{c}{Reductions} &
    \multicolumn{3}{c}{Runtime (sec.)}\\
    Dimension & \# Models &
    \# Total & \# Non-equivalent & Min. & Average & Max.
    \\
    \hline
    2 - 9 & 44& 4.02& 1.39 & 0.0 s& 0.6 s& 0.66 s\\
    10 - 19 & 41& 8.15& 2.61& 0.01 s& 0.21 s& 1.46 s\\
    20 - 29 & 46& 9.65& 2.13& 0.08 s& 0.44 s& 1.48 s\\
    30 - 39 & 17& 19.41& 2.71 & 0.33 s& 1.74 s& 5.91 s\\
    40 - 59 & 25& 29.08& 6.08 & 0.78 s& 4.58 s& 26.71 s\\
    60 - 79 & 20& 37.25& 6.95& 7.7 s& 34.57 s& 102.92 s\\
    80 - 99 & 11& 42.91& 7.09& 24.46 s& 96.38 s& 497.26 s\\
    100 - 133 & 4& 89.0& 21.5& 75.15 s& 202.52 s& 312.02 s\\
    \hline
    \end{tabular}
  \caption{Benchmark results aggregated by model dimension}
  \label{table:benchmarks}
\end{table}

The timings were produced on a laptop with 2 cores 1.60GHz each and 8 Gb RAM\footnote{For the overall table, we refer to~\url{\repository/blob/main/benchmark/biomodels_benchmark_results.md}, commit hash \texttt{23c9f532aa316cbef59a8e3e6be04156a3d9c3eb}.}.
We would like to note that out of the 208 models considered, at least one reduction was found in 202 models, and 154 of them admit a non-constant reduction. 

The timings in the table do not include the cost of the positivization step~\cite{pozitivizor}, which is optional. Here, our algorithm uses the \texttt{Polymake}~\cite{polymake:2017} library. 
With the positivization step, the running time increases no more than by a factor of two in most instances, and usually the increase is indistinguishable at all\footnote{One notable exception are models that admit large reductions with large coefficients. For example, model \texttt{BIOMD0000000153} of dimension 76 has 22 nontrivial reductions of dimensions $55$ and more, and applying the positivization routine increases the total runtime from 40 s to 1240 s.}. In the earlier versions of the implementation of~\Cref{alg:main}, computing the algebra basis on \ref{step:algmain:2} had often been a clear bottleneck on our dataset. 
With the modifications to the \Cref{alg:algebra-basis} as described in \Cref{sec:generating-algebra}, currently, the most time-consuming steps are the restriction and quotienting procedures applied on \ref{step:algmain:3} of~\Cref{alg:main}. 
Solving a number of linear systems to find the matrix representations of restricted and quotient operators is a clear bottleneck here.


\section{Case studies}\label{sec:examples}

\subsection{Inactivation of factor Va}

We will consider a model from~\cite{FactorV} which appears in the BioModels database~\cite{BioModels2020} as \texttt{BIOMD0000000365}.
Factor V is a protein involved in the process of coagulation (transforming blood from liquid to gel), and thus is closely related to blood vessel repair and thrombosis.
In particular, it can assist in activating anticoagulant protein C.
The activated factor V, factor Va, can no longer do this.
A model describing deactivation of Va by means of activated protein C (APC) was proposed and studied in~\cite{FactorV}.

Factor Va consists of the heavy chain (HC) and light chain (LC), and the binding of APC happens through the light chain.
The model consists of the following species
\begin{itemize}
    \item Factor Va and its versions Va${}_3$, Va${}_5$, Va${}_6$, Va${}_{53}$, Va${}_{56}$, Va${}_{36}$, and Va${}_{536}$;
    \item LC, HC, and the versions of the latter (HC$_{3}$, HC${}_5$, etc) corresponding to the versions of Va;
    \item the A${}_1$ domain of factor Va, Va${}_{LC\cdot A1}$ and versions of the A${}_2$ domain such as Va${}_{A3}$, Va${}_{A53}$, etc. 
    \item APC, complexes formed by it and LC/Va (such as APC$\cdot$Va${}_3$).
\end{itemize}

In total, the model contains 30 variables and 9 parameters, and the parameters are encoded as constant states as in~\Cref{rem:parameters}. Our code finds a maximal chain of lumpings of length 14 in under $5$ second on a laptop.
The smallest reduction with nonzero dynamics has dimension three and involves two parameters (similar to the one in~\Cref{ex:chain}):
\begin{equation}\label{eq:apc_reduction}
\begin{cases}
  y_1' = -k_1 y_1 y_2 + k_2 y_3,\\
  y_2' = -k_1 y_1 y_2 + k_2 y_3,\\
  y_3' = k_1 y_1 y_2 - k_2 y_3.
\end{cases}
\end{equation}
The macro-variables are
\begin{align*}
    y_1 &= [\operatorname{APC}],\\
    y_2 &= [\operatorname{LC}] + [\operatorname{Va}] + [\operatorname{Va}_3] + [\operatorname{Va}_{36}] + [\operatorname{Va}_5] + [\operatorname{Va}_{53}] + [\operatorname{Va}_{536}] + [\operatorname{Va}_{56}] + [\operatorname{Va}_{\operatorname{LC}\cdot A_1}],\\
    y_3 &= [\operatorname{LC}\cdot \operatorname{APC}] + [\operatorname{Va} \cdot \operatorname{APC}] + [\operatorname{Va}_3 \cdot \operatorname{APC}] +
    [\operatorname{Va}_{36} \cdot \operatorname{APC}] + [\operatorname{Va}_5 \cdot \operatorname{APC}]\\ &+ [\operatorname{Va}_{53} \cdot \operatorname{APC}] + [\operatorname{Va}_{536} \cdot \operatorname{APC}] + [\operatorname{Va}_{56} \cdot \operatorname{APC}] + [\operatorname{Va}_{\operatorname{LC} \cdot A_1} \cdot \operatorname{APC}].
\end{align*}
Variable $y_2$ (resp., $y_3$) can be described as the total concentration of the light chains without (resp., with) bound APC.
Therefore, the reduction~\eqref{eq:apc_reduction} focuses on the process of binding/unbinding of APC to the light chains, and it turns out that the other processes such as reactions between the heavy and light chains become irrelevant and, in particular, the HC${}_n$ species do not appear in the macro-variables at all.

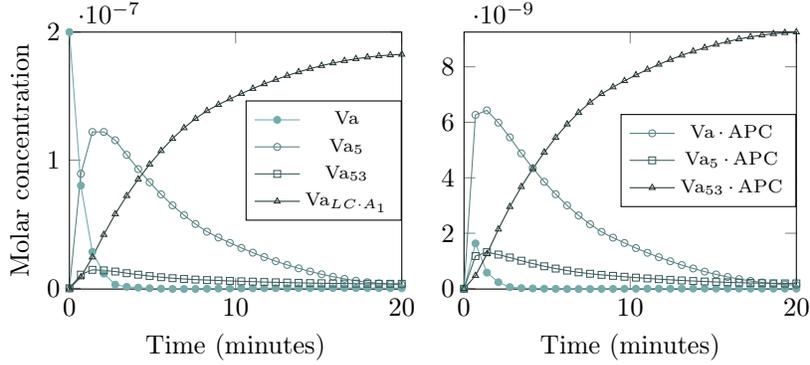
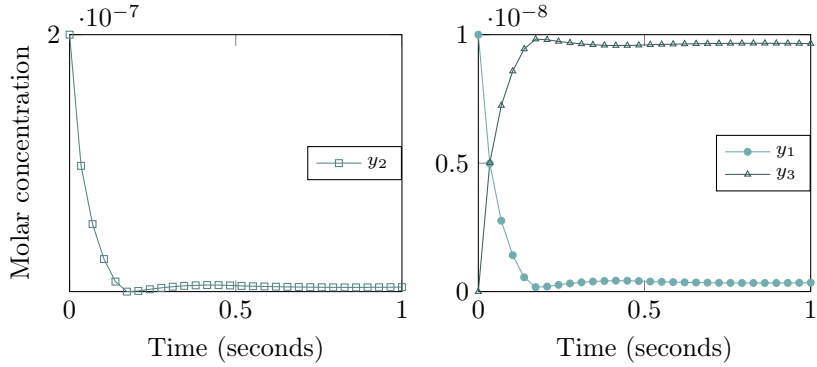
\begin{figure}[H]
\tikzset{every mark/.append style={scale=0.65}}
     \centering
     \begin{subfigure}[b]{1\textwidth}
         \centering
    \begin{tikzpicture}
\begin{axis}[
legend entries={
$\operatorname{Va}$, 
$\operatorname{Va}_5$, 
$\operatorname{Va}_{53}$, 
$\operatorname{Va}_{LC\cdot A_1}$
},
legend style={
at={(1.0,0.5)},
anchor=east,
},
xlabel=Time (minutes),
ylabel=Molar concentration,
xtick={0, 600, 1200},
xticklabels={0, 10, 20},
ytick={0, 1e-7, 2e-7},
enlargelimits=false,
colormap name=tealy,
cycle list={[colors of colormap={100, 300, 500, 700}]},
]
\addplot+[mark=*] table[x=t, y=y] {plot_data/y2_Va.dat};
\addplot+[mark=o] table[x=t, y=y] {plot_data/y2_Va5.dat};
\addplot+[mark=square] table[x=t, y=y] {plot_data/y2_Va53.dat};
\addplot+[mark=triangle] table[x=t, y=y] {plot_data/y2_VaLCA1.dat};
\end{axis}
    \end{tikzpicture}
\begin{tikzpicture}
\begin{axis}[
legend entries={, 
$\operatorname{Va}\cdot \operatorname{APC}$, 
$\operatorname{Va}_5\cdot \operatorname{APC}$, 
$\operatorname{Va}_{53}\cdot \operatorname{APC}$, 
$\operatorname{Va}_{LC\cdot A_1}\cdot \operatorname{APC}$
},
legend style={
at={(1.0,0.5)},
anchor=east,
},
xlabel=Time (minutes),
xtick={0, 600, 1200},
xticklabels={0, 10, 20},
enlargelimits=false,
colormap name=tealy,
cycle list={[colors of colormap={100, 300, 500, 700}]},
]
\addplot+[mark=*] table[x=t, y=y] {plot_data/y3_Va_APC.dat};
\addplot+[mark=o] table[x=t, y=y] {plot_data/y3_Va5_APC.dat};
\addplot+[mark=square] table[x=t, y=y] {plot_data/y3_Va53_APC.dat};
\addplot+[mark=triangle] table[x=t, y=y] {plot_data/y3_VaLCA1_APC.dat};
\end{axis}
    \end{tikzpicture}
         \caption{Some of the states of the original model appearing in $y_2$ and $y_3$}
         \label{fig:apc_original}
     \end{subfigure}
\hfill

     \begin{subfigure}[b]{1\textwidth}
     \centering
            \begin{tikzpicture}
\begin{axis}[
legend entries={$y_2$},
legend style={
at={(1.0,0.5)},
anchor=east,
},
xlabel=Time (seconds),
ylabel=Molar concentration,
xtick={0, 0.5, 1.0},
ytick={0, 1e-7, 2e-7},
enlargelimits=false,
colormap name=tealy,
cycle list={[colors of colormap={300}]},
]
\addplot+[mark=square] table[x=t, y=y] {plot_data/y2.dat};
\end{axis}
    \end{tikzpicture}
\begin{tikzpicture}
\begin{axis}[
legend entries={$y_1$, $y_3$},
legend style={
at={(1.0,0.5)},
anchor=east,
},
xlabel=Time (seconds),
xtick={0, 0.5, 1.0},
ytick={0, 0.5e-8, 1e-8},
enlargelimits=false,
colormap name=tealy,
cycle list={[colors of colormap={100, 500}]},
]
\addplot+[mark=*] table[x=t, y=y] {plot_data/y1.dat};
\addplot+[mark=triangle] table[x=t, y=y] {plot_data/y3.dat};
\end{axis}
    \end{tikzpicture}
    \caption{States of the reduced model}
    \label{fig:apc_reduced}
     \end{subfigure}
        \caption{Numerical simulation for the model from~\cite{FactorV} and its reduction using the initial conditions and parameter values from~\cite{FactorV}}
        \label{fig:apc}
\end{figure}

From a numerical perspective\footnote{All numerical simulations in this paper have been done using \texttt{ModelingToolkit}~\cite{ma2021modelingtoolkit} and \texttt{DifferentialEquations.jl}~\cite{rackauckas2017differentialequations}}, the reduction~\eqref{eq:apc_reduction} can be 
interpreted as exact timescale separation since the dynamics of the macro-variables turns out to be transient compared to the dynamics of the original system.
More precisely, the original system was studied in~\cite{FactorV} and has nontrivial dynamics on the timespan of 1200 second.
In particular, this is the case for the variables contributing to the macro-variable $y_2$, see~\Cref{fig:apc_original}.
On the other hand, as \Cref{fig:apc_reduced} shows, the macro-variables $y_1, y_2, y_3$ have much faster dynamics and reach the steady state after less than one second.


\subsection{Model of cell death}\label{sec:ex_celldeath}

In this subsection, we consider a model designed in~\cite{celldeath} in order to study the sensitivity of the apoptosis (programmed cell death) to the TNF (tumor necrosis factor) stimulation.
The overall model involves 47 chemical species and numerous interactions between them schematically described in~\cite[Figure~1]{celldeath}.
Our code produces a maximal chain of lumpings of length 23 (16 out of them with nonconstant dynamics).

We will consider the nonconstant reduction of the smallest dimension.
It involves two proteins, A20 and FLIP, whose concentrations depend on the concentrations of the corresponding mRNAs, A20\_mRNA and FLIP\_mRNA.
The concentrations of these mRNAs are governed by the concentrations of nuclear NF-$\kappa$B (NFkB\_N).
The latter depends (directly or indirectly) on many other species including the aforementioned protein A20.

\begin{figure}
\centering
\begin{tikzpicture}
  \node[draw, text width=8em,align=center,fill=mygolden] (complex) at (0,0) {\footnotesize 42 other species,\\including the ligand-receptor complex};
  \node (nfkb) at (0, -2) {\small NFkB};
  \node (a20mrna) at (2.5, -2) {\small A20 mRNA};
  \node (a20) at (2.5, -1) {\small A20};
  \node (flipmrna) at (-2.5, -2) {\small FLIP mRNA};
 \node (flip) at (-2.5, -1) {\small FLIP};

\draw [arrowstyle] (nfkb) -- (a20mrna);
\draw [arrowstyle] (a20mrna) -- (a20);
\draw [arrowstyle] (nfkb) -- (flipmrna);
\draw [arrowstyle] (flipmrna) -- (flip);

\draw [arrowstyle, transform canvas={xshift=2pt}, shorten <=0.2em, shorten >=0.2em] (complex) -- (nfkb);
\draw [arrowstyle, transform canvas={xshift=-2pt}, shorten <=0.2em, shorten >=0.2em] (nfkb) -- (complex);

\draw [arrowstyle, shorten <=-0.3em, shorten >=0.2em] (a20) -- (complex);
\draw [arrowstyle, shorten <=-0.3em, shorten >=0.2em] (flip) -- (complex);

\end{tikzpicture}
\caption{The relevant chemical species and dependencies between them}
\label{fig:ex_celldeath2}
\end{figure}
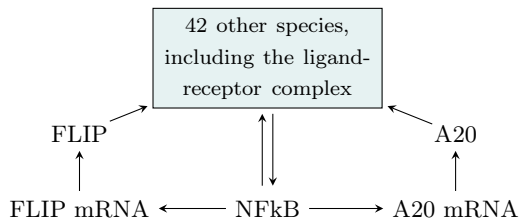

These species and relations between them are summarized on~\Cref{fig:ex_celldeath2}, and the corresponding differential equations are:
\begin{align*}
    &[A20]' = k_1 [A20\_mRNA] + k_2, & &[A20\_mRNA]' = k_5 [NF{\kappa}B\_N],\\
    &[FLIP]' = k_3 [FLIP\_mRNA] + k_4, & &[FLIP\_mRNA]' = k_6 [NF{\kappa}B\_N],
\end{align*}
where $k_1, \ldots, k_6$ are numeric parameters.
Our code finds a three-dimensional reduction which can be straightforwardly simplified further a two-dimensional with the following macro-variables $y_1, y_2$ and the reduced system:
\[
\begin{cases}
  y_1 = \frac{k_6}{k_1} [A20] - \frac{k_5}{k_3} [FLIP],\\
  y_2 = k_6 [A20\_mRNA] - k_5 [FLIP\_mRNA]
\end{cases}
\implies
\begin{cases}
  y_1' = y_2 + \frac{k_2k_6}{k_1} - \frac{k_4 k_5}{k_3},\\
  y_2' = 0
\end{cases}
\]
So the idea is that, although both A20 and FLIP are involved in a complex reaction network, one can, by eliminating the dependence on NF$\kappa$B, find a linear combination of them satisfying a simple system of differential equations which can be explicitly solved. 
Such explicit relations on the states can be, for example, combined with the differential inequalities method in order to obtain tighter reachability bounds~\cite{SCOTT201393}. 

By going further along the chain of the reductions one can include gradually more species into the reduced model, for example, a combination of the RIP protein and the transitional receptor can be included in a similar fashion.
%
%
\section{Conclusions}

We have presented a new algorithm which takes as input a system of ODEs and produces a longest possible chain of exact linear reductions of the system such that each reduction in the chain is a refinement of the previous one.
This specification is more general compared to the existing tools as it does not put any restriction on the new variables other than being the linear combinations of the original ones and it does not require any initial observable/guess.
We expect that our approach can be extended to the systems with the rational right-hand side using the ideas from~\cite{JimnezPastor2022}, we leave this for future research.

We provided a publicly available implementation in Julia.
Our code is able to analyze models of dimension over a hundred in a couple of minutes using commodity hardware.
We have also demonstrated its applicability to models arising in life sciences.
The performance can be further improved, for example, by first searching for linear first integrals (which appear frequently, for example, in chemical reaction networks) and using these constant reductions to skip some of the steps of our algorithm.

Since the produced reductions are exact, our tool can be used for formal verification and as a preprocessing for approximate reduction techniques.
While exactness is thus an important feature, it can also be viewed as a limitation since some models have only a few exact reductions (if any).
Therefore, one intriguing direction for future research is to produce a 
relaxed version of our algorithm to find approximate lumpings together with rigorous error bounds.
For existing results on approximate lumping and related approaches based on the singular perturbation theory, see~\cite{approx_lump1,approx_lump2,Soliman2014} and references therein.
Interestingly, the core linear algebraic problem of our algorithm, finding common invariant subspaces, has been recently studied from the perspective of approximate but rigorous computation in~\cite{factor1,factor2} motivated by factoring linear differential operators.
We expect the ideas from these papers to be useful in our context as well.

\section*{Acknowledgments}
We would like to thank David E. Speyer for his clear and detailed note~\cite{speyer}.
We would like to thank Mirco Tribastone for helpful discussions and Rongwei Yang for discussions about~\Cref{prop:chpoly}.
We would like to thank the referees for helpful feedback.
GP was supported by the Paris Ile-de-France region (project ``XOR'') and partially supported by NSF grants DMS-1853482, DMS-1760448, and DMS-1853650.
AD was supported by the Max Planck Institute for Informatics.

\bibliographystyle{elsarticle-num} 
\bibliography{mybibliography}

\begin{thebibliography}{10}
\expandafter\ifx\csname url\endcsname\relax
  \def\url#1{\texttt{#1}}\fi
\expandafter\ifx\csname urlprefix\endcsname\relax\def\urlprefix{URL }\fi
\expandafter\ifx\csname href\endcsname\relax
  \def\href#1#2{#2} \def\path#1{#1}\fi

\bibitem{FeretPNAS}
J.~Feret, V.~Danos, J.~Krivine, R.~Harmer, W.~Fontana,
  \href{http://dx.doi.org/10.1073/pnas.0809908106}{Internal coarse-graining of
  molecular systems}, Proceedings of the National Academy of Sciences 106~(16)
  (2009) 6453--6458.
\newline\urlprefix\url{http://dx.doi.org/10.1073/pnas.0809908106}

\bibitem{feret2}
J.~Feret, T.~Henzinger, H.~Koeppl, T.~Petrov,
  \href{https://doi.org/10.1016/j.tcs.2011.12.059}{Lumpability abstractions of
  rule-based systems}, Theoretical Computer Science 431 (2012) 137--164.
\newline\urlprefix\url{https://doi.org/10.1016/j.tcs.2011.12.059}

\bibitem{erode}
L.~Cardelli, M.~Tribastone, M.~Tschaikowski, A.~Vandin,
  \href{https://doi.org/10.1007/978-3-662-54580-5\_19}{{ERODE}: A tool for the
  evaluation and reduction of ordinary differential equations}, in: TACAS 2017,
  Vol. 10206 of LNCS, 2017.
\newline\urlprefix\url{https://doi.org/10.1007/978-3-662-54580-5\_19}

\bibitem{Cardelli2017a}
L.~Cardelli, M.~Tribastone, M.~Tschaikowski, A.~Vandin,
  \href{https://www.pnas.org/content/114/38/10029}{Maximal aggregation of
  polynomial dynamical systems}, Proceedings of the National Academy of
  Sciences 114~(38) (2017) 10029--10034.
\newline\urlprefix\url{https://www.pnas.org/content/114/38/10029}

\bibitem{clue}
A.~Ovchinnikov, I.~P. Verona, G.~Pogudin, M.~Tribastone,
  \href{https://doi.org/10.1093/bioinformatics/btab010}{{CLUE}: exact maximal
  reduction of kinetic models by constrained lumping of differential
  equations}, Bioinformatics (2021).
\newline\urlprefix\url{https://doi.org/10.1093/bioinformatics/btab010}

\bibitem{antoulas}
A.~Antoulas, Approximation of Large-Scale Dynamical Systems, Adv. in Design and
  Control, SIAM, 2005.

\bibitem{Hubert2013}
E.~Hubert, G.~Labahn, \href{https://doi.org/10.1007/s10208-013-9165-9}{Scaling
  invariants and symmetry reduction of dynamical systems}, Foundations of
  Computational Mathematics 13~(4) (2013) 479--516.
\newline\urlprefix\url{https://doi.org/10.1007/s10208-013-9165-9}

\bibitem{Camporesi2011}
F.~Camporesi, J.~Feret,
  \href{https://doi.org/10.1016/j.entcs.2011.09.014}{Formal reduction for
  rule-based models}, Electronic Notes in Theoretical Computer Science 276
  (2011) 29--59.
\newline\urlprefix\url{https://doi.org/10.1016/j.entcs.2011.09.014}

\bibitem{Olver1995}
P.~J. Olver, \href{https://doi.org/10.1017/cbo9780511609565}{Equivalence,
  Invariants and Symmetry}, Cambridge University Press, 1995.
\newline\urlprefix\url{https://doi.org/10.1017/cbo9780511609565}

\bibitem{Sankaranarayanan2011}
S.~Sankaranarayanan, \href{https://doi.org/10.1145/1967701.1967723}{Automatic
  abstraction of non-linear systems using change of bases transformations}, in:
  Proceedings of the 14th international conference on Hybrid systems:
  computation and control, 2011.
\newline\urlprefix\url{https://doi.org/10.1145/1967701.1967723}

\bibitem{Faeder2005}
J.~R. Faeder, M.~L. Blinov, W.~S. Hlavacek,
  \href{https://doi.org/10.1145/1066677.1066712}{Graphical rule-based
  representation of signal-transduction networks}, in: Proceedings of the 2005
  {ACM} symposium on Applied computing, 2005.
\newline\urlprefix\url{https://doi.org/10.1145/1066677.1066712}

\bibitem{Danos}
V.~Danos, J.~Feret, W.~Fontana, R.~Harmer, J.~Krivine,
  \href{https://doi.org/10.1007/978-3-540-74407-8_3}{Rule-based modelling of
  cellular signalling}, in: {CONCUR} 2007 {\textendash} Concurrency Theory,
  2007, pp. 17--41.
\newline\urlprefix\url{https://doi.org/10.1007/978-3-540-74407-8_3}

\bibitem{Camporesi2013}
F.~Camporesi, J.~Feret, J.~Hayman,
  \href{https://doi.org/10.1007/978-3-642-40708-6_17}{Context-sensitive flow
  analyses: A hierarchy of model reductions}, in: Computational Methods in
  Systems Biology, 2013, pp. 220--233.
\newline\urlprefix\url{https://doi.org/10.1007/978-3-642-40708-6_17}

\bibitem{feret:CMSB2017}
F.~Camporesi, J.~Feret, K.~Q. L{\'y}, Ka{DE}: a tool to compile {K}appa rules
  into (reduced) {ODE}s models, in: Fifteenth International Workshop on Static
  Analysis and Systems Biology (SASB'17), Vol. 10545 of LNCS/LNBI, springer,
  2017.

\bibitem{Cardelli2019}
L.~Cardelli, M.~Tribastone, M.~Tschaikowski, A.~Vandin,
  \href{https://doi.org/10.1016/j.tcs.2019.03.018}{Symbolic computation of
  differential equivalences}, Theoretical Computer Science 777 (2019) 132--154.
\newline\urlprefix\url{https://doi.org/10.1016/j.tcs.2019.03.018}

\bibitem{JimnezPastor2022}
A.~Jim{\'{e}}nez-Pastor, J.~P. Jacob, G.~Pogudin,
  \href{https://doi.org/10.1007/978-3-031-15034-0_10}{Exact linear reduction
  for~rational dynamical systems}, in: Computational Methods in Systems
  Biology, Springer International Publishing, 2022, pp. 198--216.
\newline\urlprefix\url{https://doi.org/10.1007/978-3-031-15034-0_10}

\bibitem{LiRabitz}
G.~Li, H.~Rabitz, \href{https://doi.org/10.1016/0009-2509(89)85014-6}{A general
  analysis of exact lumping in chemical kinetics}, Chemical Engineering Science
  44~(6) (1989) 1413--1430.
\newline\urlprefix\url{https://doi.org/10.1016/0009-2509(89)85014-6}

\bibitem{RONYAI}
L.~R\'onyai, \href{https://doi.org/10.1016/S0747-7171(08)80017-X}{Computing the
  structure of finite algebras}, Journal of Symbolic Computation 9~(3) (1990)
  355--373.
\newline\urlprefix\url{https://doi.org/10.1016/S0747-7171(08)80017-X}

\bibitem{CGK}
A.~Chistov, G.~Ivanyos, M.~Karpinski,
  \href{https://doi.org/10.1145/258726.258751}{Polynomial time algorithms for
  modules over finite dimensional algebras}, in: Proceedings of the 1997
  International Symposium on Symbolic and Algebraic Computation, 1997, p.
  68–74.
\newline\urlprefix\url{https://doi.org/10.1145/258726.258751}

\bibitem{speyer}
D.~Speyer,
  \href{https://mathematica.stackexchange.com/questions/6519/is-there-a-clean-way-to-extract-the-subspaces-invariant-under-a-list-of-matrices/9442#9442}{Response
  to ``{I}s there a clean way to extract the subspaces invariant under a list
  of matrices?''}.
\newline\urlprefix\url{https://mathematica.stackexchange.com/questions/6519/is-there-a-clean-way-to-extract-the-subspaces-invariant-under-a-list-of-matrices/9442#9442}

\bibitem{BioModels2020}
R.~S. Malik-Sheriff, M.~Glont, T.~V.~N. Nguyen, K.~Tiwari, M.~G. Roberts,
  A.~Xavier, M.~T. Vu, J.~Men, M.~Maire, S.~Kananathan, E.~L. Fairbanks, J.~P.
  Meyer, C.~Arankalle, T.~M. Varusai, V.~Knight-Schrijver, L.~Li,
  C.~Dueñas-Roca, G.~Dass, S.~M. Keating, Y.~M. Park, N.~Buso, N.~Rodriguez,
  M.~Hucka, H.~Hermjakob, \href{https://doi.org/10.1093/nar/gkz1055}{{BioModels
  — 15 years of sharing computational models in life science}}, Nucleic Acids
  Research 48~(D1) (2020) D407--D415.
\newline\urlprefix\url{https://doi.org/10.1093/nar/gkz1055}

\bibitem{pozitivizor}
G.~Pogudin, X.~Zhang,
  \href{https://doi.org/10.1007/978-3-030-85633-5\_6}{Interpretable exact
  linear reductions via positivity}, in: E.~Cinquemani, L.~Paulev{\'e} (Eds.),
  Computational Methods in Systems Biology, 2021, pp. 91--107.
\newline\urlprefix\url{https://doi.org/10.1007/978-3-030-85633-5\_6}

\bibitem{mass_action}
S.~Dunn, A.~Constantinides, P.~Moghe,
  \href{https://doi.org/10.1016/B978-0-12-186031-8.X5000-6}{Numerical Methods
  in Biomedical Engineering}, Academic Press, 2006.
\newline\urlprefix\url{https://doi.org/10.1016/B978-0-12-186031-8.X5000-6}

\bibitem{Feinberg}
M.~Feinberg, \href{https://doi.org/10.1007/978-3-030-03858-8}{Foundations of
  Chemical Reaction Network Theory}, Springer Cham, 2019.
\newline\urlprefix\url{https://doi.org/10.1007/978-3-030-03858-8}

\bibitem{every_variety}
M.~Reineke, \href{https://doi.org/10.1007/s10468-012-9357-z}{Every projective
  variety is a quiver {G}rassmannian}, Algebras and Representation Theory 16
  (2013) 1313--1314.
\newline\urlprefix\url{https://doi.org/10.1007/s10468-012-9357-z}

\bibitem{quivers}
E.~Nijholt, B.~W. Rink, S.~Schwenker,
  \href{https://doi.org/10.1137/20m1345670}{Quiver representations and
  dimension reduction in dynamical systems}, {SIAM} Journal on Applied
  Dynamical Systems 19~(4) (2020) 2428--2468.
\newline\urlprefix\url{https://doi.org/10.1137/20m1345670}

\bibitem{kir_drozd}
Y.~A. Drozd, V.~V. Kirichenko, Finite Dimensional Algebras, Springer-Verlag,
  1994.

\bibitem{constructive_survey}
M.~R. Bremner, \href{https://doi.org/10.1515/gcc.2011.003}{How to compute the
  {W}edderburn decomposition of a finite-dimensional associative algebra},
  Groups, Complexity, Cryptology 3~(1) (2011) 47--66.
\newline\urlprefix\url{https://doi.org/10.1515/gcc.2011.003}

\bibitem{KlepVolvic}
I.~Klep, J.~Vol{\v{c}}i{\v{c}}, A note on group representations, determinantal
  hypersurfaces and their quantizations, in: M.~A. Bastos, L.~Castro, A.~Y.
  Karlovich (Eds.), Operator Theory, Functional Analysis and Applications,
  Springer International Publishing, 2021, pp. 393--402.

\bibitem{Cohen}
S.~D. Cohen, \href{https://doi.org/10.1112/plms/s3-43.2.227}{The distribution
  of {G}alois groups and {H}ilbert's irreducibility theorem}, Proceedings of
  the London Mathematical Society s3-43~(2) (1981) 227--250.
\newline\urlprefix\url{https://doi.org/10.1112/plms/s3-43.2.227}

\bibitem{Zippel}
R.~Zippel, \href{http://dx.doi.org/10.1007/978-1-4615-3188-3}{Effective
  Polynomial Computation}, Springer, 1993.
\newline\urlprefix\url{http://dx.doi.org/10.1007/978-1-4615-3188-3}

\bibitem{linalg}
S.~Axler, \href{https://doi.org/10.1007/978-3-319-11080-6}{Linear Algebra Done
  Right}, Springer Cham, 2015.
\newline\urlprefix\url{https://doi.org/10.1007/978-3-319-11080-6}

\bibitem{bezanson2017julia}
J.~Bezanson, A.~Edelman, S.~Karpinski, V.~B. Shah,
  \href{https://doi.org/10.1137/141000671}{Julia: A fresh approach to numerical
  computing}, SIAM review 59~(1) (2017) 65--98.
\newline\urlprefix\url{https://doi.org/10.1137/141000671}

\bibitem{Nemo.jl-2017}
C.~Fieker, W.~Hart, T.~Hofmann, F.~Johansson,
  \href{https://doi.acm.org/10.1145/3087604.3087611}{Nemo/hecke: Computer
  algebra and number theory packages for the julia programming language}, in:
  Proceedings of the 2017 ACM on International Symposium on Symbolic and
  Algebraic Computation, ISSAC '17, ACM, New York, NY, USA, 2017, pp. 157--164.
\newline\urlprefix\url{https://doi.acm.org/10.1145/3087604.3087611}

\bibitem{Hart2010}
W.~B. Hart, Fast library for number theory: An introduction, in: Proceedings of
  the Third International Congress on Mathematical Software, ICMS'10,
  Springer-Verlag, Berlin, Heidelberg, 2010, pp. 88--91,
  \url{https://flintlib.org}.

\bibitem{Calcium-arxiv}
F.~Johansson, \href{https://doi.org/10.1145/3452143.3465513}{Calcium}, in:
  Proceedings of the 2021 on International Symposium on Symbolic and Algebraic
  Computation, 2021.
\newline\urlprefix\url{https://doi.org/10.1145/3452143.3465513}

\bibitem{suitesparse}
T.~A. Davis, \href{https://doi.org/10.1145/3322125}{Algorithm 1000:
  Suitesparse:graphblas: Graph algorithms in the language of sparse linear
  algebra}, ACM Trans. Math. Softw. 45~(4) (dec 2019).
\newline\urlprefix\url{https://doi.org/10.1145/3322125}

\bibitem{erode_manual}
\href{https://www.erode.eu/docs/ERODE-manual.pdf}{{ERODE}: {E}valuation and
  {R}eduction of {S}tochastic {R}eaction {N}etworks and {D}ifferential
  {E}quations}, user's manual (2020).
\newline\urlprefix\url{https://www.erode.eu/docs/ERODE-manual.pdf}

\bibitem{largescale}
I.~C. P{\'e}rez-Verona, M.~Tribastone, A.~Vandin, A large-scale assessment of
  exact model reduction in the {B}io{M}odels repository, in: L.~Bortolussi,
  G.~Sanguinetti (Eds.), Computational Methods in Systems Biology, Springer
  International Publishing, 2019, pp. 248--265.

\bibitem{polymake:2017}
B.~Assarf, E.~Gawrilow, K.~Herr, M.~Joswig, B.~Lorenz, A.~Paffenholz, T.~Rehn,
  \href{http://dx.doi.org/10.1007/s12532-016-0104-z}{Computing convex hulls and
  counting integer points with {\tt polymake}}, Math. Program. Comput. 9~(1)
  (2017) 1--38.
\newline\urlprefix\url{http://dx.doi.org/10.1007/s12532-016-0104-z}

\bibitem{FactorV}
M.~F. Hockin, K.~M. Cawthern, M.~Kalafatis, K.~G. Mann,
  \href{https://doi.org/10.1021/bi981966e}{A model describing the inactivation
  of factor {V}a by {APC}: Bond cleavage, fragment dissociation, and product
  inhibition}, Biochemistry 38~(21) (1999) 6918--6934.
\newline\urlprefix\url{https://doi.org/10.1021/bi981966e}

\bibitem{ma2021modelingtoolkit}
Y.~Ma, S.~Gowda, R.~Anantharaman, C.~Laughman, V.~Shah, C.~Rackauckas,
  Modeling{T}oolkit: A composable graph transformation system for
  equation-based modeling (2021).
\newblock \href {http://arxiv.org/abs/2103.05244} {\path{arXiv:2103.05244}}.

\bibitem{rackauckas2017differentialequations}
C.~Rackauckas, Q.~Nie, Differential{E}quations.jl--a performant and
  feature-rich ecosystem for solving differential equations in {J}ulia, Journal
  of Open Research Software 5~(1) (2017).

\bibitem{celldeath}
M.~Schliemann, E.~Bullinger, S.~Borchers, F.~Allg\"ower, R.~Findeisen,
  P.~Scheurich, \href{https://doi.org/10.1186/1752-0509-5-204}{Heterogeneity
  reduces sensitivity of cell death for {TNF}-stimuli}, BMC Systems Biology
  5~(204) (2011).
\newline\urlprefix\url{https://doi.org/10.1186/1752-0509-5-204}

\bibitem{SCOTT201393}
J.~K. Scott, P.~I. Barton,
  \href{https://doi.org/10.1016/j.automatica.2012.09.020}{Bounds on the
  reachable sets of nonlinear control systems}, Automatica 49~(1) (2013)
  93--100.
\newline\urlprefix\url{https://doi.org/10.1016/j.automatica.2012.09.020}

\bibitem{approx_lump1}
M.~Tschaikowski, M.~Tribastone,
  \href{https://doi.org/10.1109/TAC.2015.2457172}{Approximate reduction of
  heterogenous nonlinear models with differential hulls}, IEEE Transactions on
  Automatic Control 61~(4) (2016) 1099--1104.
\newline\urlprefix\url{https://doi.org/10.1109/TAC.2015.2457172}

\bibitem{approx_lump2}
A.~Girard, G.~J. Pappas,
  \href{https://doi.org/10.3166/ejc.17.568-578}{Approximate bisimulation: A
  bridge between computer science and control theory}, European Journal of
  Control 17~(5) (2011) 568--578.
\newline\urlprefix\url{https://doi.org/10.3166/ejc.17.568-578}

\bibitem{Soliman2014}
S.~Soliman, F.~Fages, O.~Radulescu,
  \href{https://doi.org/10.1186/s13015-014-0024-2}{A constraint solving
  approach to model reduction by tropical equilibration}, Algorithms for
  Molecular Biology 9~(1) (2014).
\newline\urlprefix\url{https://doi.org/10.1186/s13015-014-0024-2}

\bibitem{factor1}
A.~Goyer, \href{https://doi.org/10.1145/3493492.3493496}{A {S}age package for
  the symbolic-numeric factorization of linear differential operators}, ACM
  Commun. Comput. Algebra 55~(2) (2021) 44–48.
\newline\urlprefix\url{https://doi.org/10.1145/3493492.3493496}

\bibitem{factor2}
F.~Chyzak, A.~Goyer, M.~Mezzarobba,
  \href{https://doi.org/10.1145/3476446.3535503}{Symbolic-numeric factorization
  of differential operators}, in: Proceedings of the 2022 International
  Symposium on Symbolic and Algebraic Computation, ISSAC '22, 2022, p. 73–82.
\newline\urlprefix\url{https://doi.org/10.1145/3476446.3535503}

\end{thebibliography}

\end{document}